\documentclass{article}

\usepackage{microtype}
\usepackage{graphicx}
\usepackage{subfigure}
\usepackage{booktabs} 
\usepackage{fancyhdr} 
\usepackage{lastpage} 
\usepackage{extramarks} 
\usepackage{graphicx} 
\usepackage{lipsum} 
\usepackage{bm}
\usepackage{mathtools}
\usepackage{amssymb}
\usepackage{amsthm}
\usepackage{hyperref}

\newtheorem{theorem}{Theorem}[section]
\newtheorem{lemma}[theorem]{Lemma}

\newtheorem{example}[theorem]{Example}
\newtheorem{proposition}[theorem]{Proposition}
\newtheorem{corollary}[theorem]{Corollary}
\newtheorem{assumption}[theorem]{Assumption}
\theoremstyle{definition}
\newtheorem{definition}[theorem]{Definition}

\DeclareMathOperator{\pa}{pa}
\DeclareMathOperator{\neigh}{ne}
\DeclareMathOperator{\an}{an}

\makeatletter
\newcommand*{\indep}{%
  \mathbin{%
    \mathpalette{\@indep}{}%
  }%
}
\newcommand*{\nindep}{%
  \mathbin{
    \mathpalette{\@indep}{\not}
  }%
}
\newcommand*{\@indep}[2]{%
  \sbox0{$#1\perp\m@th$}
  \sbox2{$#1=$}
  \sbox4{$#1\vcenter{}$}
  \rlap{\copy0}
  \dimen@=\dimexpr\ht2-\ht4-.2pt\relax
  \kern\dimen@
  {#2}%
  \kern\dimen@
  \copy0 
} 
\makeatother

\usepackage[accepted]{icml2018}

\icmltitlerunning{Characterizing and Learning Equivalence Classes of Causal DAGs under Interventions}

\begin{document}

\twocolumn[
\icmltitle{                                                                                                                                                                                                                                                                                                                                                                Characterizing and Learning Equivalence Classes of Causal DAGs\\ under Interventions
}

\begin{icmlauthorlist}
\icmlauthor{Karren D. Yang}{mit}
\icmlauthor{Abigail Katcoff}{mit}
\icmlauthor{Caroline Uhler}{mit}
\end{icmlauthorlist}

\icmlaffiliation{mit}{Massachusetts Institute of Technology, Cambridge, MA}

\icmlcorrespondingauthor{Karren Yang, Caroline Uhler}{\{karren, cuhler\}@mit.edu}


\vskip 0.3in
]

\printAffiliationsAndNotice{} 

\begin{abstract}
We consider the problem of learning causal DAGs in the setting where both observational and interventional data is available. This setting is common in biology, where gene regulatory networks can be intervened on using chemical reagents or gene deletions.
\citet{hauser12} previously characterized the identifiability of causal DAGs under perfect interventions, which eliminate dependencies between targeted variables and their direct causes. 
In this paper, we extend these identifiability results to \emph{general interventions}, which may modify the dependencies between targeted variables and their causes without eliminating them. We define and characterize the \emph{interventional Markov equivalence class} that can be identified from general (not necessarily perfect) intervention experiments. We also propose the first provably consistent algorithm for learning DAGs in this setting and evaluate our algorithm on simulated and biological datasets.

\end{abstract}

\section{Introduction}

The problem of learning a causal \emph{directed acyclic graph} (DAG) from observational data over its nodes is important across disciplines such as computational biology, sociology, and economics \cite{friedman00,pearl03,robins00,spirtes00}. A causal DAG imposes conditional independence (CI) relations on its node variables that can be used to infer its structure. Since multiple DAGs can encode the same CI relations, a causal DAG is generally only identifiable up to its \emph{Markov equivalence class} (MEC) \cite{verma90,andersson97}.

The identifiability of causal DAGs can be improved by performing \emph{interventions} on the variables. Interventions that eliminate the dependency between targeted variables and their causes are known as \emph{perfect} (or \emph{hard}) interventions \cite{eberhardt05}.
Under perfect interventions, the identifiability of causal DAGs improves to a smaller equivalence class called the \emph{perfect-$\mathcal{I}$-MEC}%
\footnote{In \citet{hauser12}, they call this the interventional MEC ($\mathcal{I}$-MEC). We call it the perfect-$\mathcal{I}$-MEC to avoid confusion with the equivalence class for DAGs under general interventions that we characterize in this paper, which we call the $\mathcal{I}$-MEC.}
\cite{hauser12}. Recently, \citet{wang17} proposed the first provably consistent algorithm for recovering the perfect-$\mathcal{I}$-MEC and successfully applied it towards learning regulatory networks from interventional data.

However, only considering perfect interventions is restrictive: in practice, many interventions are \emph{non-perfect} (or \emph{soft}) and modify the causal relations between targeted variables and their direct causes without eliminating them \cite{eberhardt05}. 
In genomics, for example, interventions such as RNA interference or CRISPR-mediated gene activation often have only modest effects on gene suppression and activation respectively \cite{dominguez16}. Even interventions meant to be perfect, such as CRISPR/Cas9-mediated gene deletions, may not be uniformly successful across a cell population \cite{dixit16}. Although non-perfect interventions may be considered inefficient from an engineering perspective, they may still provide valuable information about regulatory networks. The identifiability of causal DAGs in this setting needs to be formally analyzed to develop maximally effective algorithms for learning from these types of interventions.
 
In this paper, we define and characterize \emph{$\mathcal{I}$-Markov equivalence classes ($\mathcal{I}$-MECs)} of causal DAGs that can be identified from \emph{general}  interventions that are not assumed to be perfect, thus extending the results of \citet{hauser12} (Section~\ref{sec_ident}). 
We show that under reasonable assumptions on the experiments, general interventions provide the same causal information as perfect interventions. These insights allow us to develop the first \emph{provably consistent algorithm} for learning the $\mathcal{I}$-MEC from data from general interventions (Section~\ref{sec_alg}), which we evaluate on synthetic and biological datasets (Section~\ref{sec_res}).

\section{Related Work} \label{sec:related-work}

\subsection{Identifiability of causal DAGs}
Given only observational data and without further distributional assumptions%
\footnote{See \citet{shimizu06}, \citet{Hoyer09}, \citet{peters14} for identifiability results for non-Gaussian or nonlinear structural equation models.}%
, the identifiability of a causal DAG is limited to its MEC \cite{verma90}.
\citet{hauser12} proved that a smaller class of DAGs, the perfect-$\mathcal{I}$-MEC, can be identified given data from perfect interventions. They conjectured but did not prove that their results extend to soft interventions.
For general interventions, \citet{tian01} presented a graph-based criterion for two DAGs being indistinguishable under single-variable interventions. Their criterion is consistent with Hauser and B\"uhlmann's perfect-$\mathcal{I}$-MEC, but they did not discuss equivalence classes, nor did they consider multi-variable interventions. \citet{eberhardt07} and \citet{eberhardt08} provided results on the number of single-target interventions required for full identifiability of the causal DAG. However, their work does not characterize equivalence classes for when the DAG is only partially identifiable.

\subsection{Causal inference algorithms}
There are two main categories of algorithms for learning causal graphs from observational data: \emph{constraint-based} and \emph{score-based} \cite{brown05, murphy01}. Constraint-based algorithms, such as the prominent PC algorithm \cite{spirtes00}, view causal inference as a constraint satisfaction problem based on CI relations inferred from data. Score-based algorithms, such as greedy equivalence search (GES) \cite{chickering02}, maximize a particular score function over the space of graphs. Hybrid algorithms such as greedy sparsest permutation (GSP) combine elements of both methods \cite{solus17}.

Algorithms have also been developed to learn causal graphs from both observational and interventional data. GIES is an extension of GES that incorporates interventional data into the score function it uses to search over the space of DAGs \cite{hauser12}, but it is in general not consistent \cite{wang17}. \emph{Perfect interventional GSP} (perfect-IGSP) is a provably consistent extension of GSP that uses interventional data to reduce the search space and orient edges, but it requires perfect interventions \cite{wang17}. 
Methods that allow for latent confounders and unknown intervention targets include \citet{eaton07}, JCI \cite{magliacane16}, HEJ \cite{hyttinen14}, CombINE \cite{triantafillou15}, and ICP \cite{peters16}, but they do not have consistency guarantees for returning a DAG in the correct class.

\section{Identifiability under general interventions}
\label{sec_ident}
In this section, we characterize the $\mathcal{I}$-MEC: a smaller equivalence class than the MEC that can be identified under general interventions with known targets. The main result is a graphical criterion for determining whether two DAGs are $\mathcal{I}$-Markov equivalent, which extends the identifiability results of \citet{hauser12} from perfect interventions to general interventions.

\subsection{Preliminaries}
Let the causal DAG $\mathcal{G} = ([p],E)$ represent a causal model in which every node $i \in [p]$ is associated with a random variable $X_i$, and let $f$ denote the joint probability distribution over $X = (X_1, \cdots, X_p)$. 
Under the {causal Markov assumption}, $f$ satisfies the \emph{Markov property} (or \emph{is Markov}) with respect to $\mathcal{G}$, i.e., $f(X) = \prod_{i} f(X_i | X_{\pa_{\mathcal{G}}(i)})$, where $\pa_{\mathcal{G}}(i)$ denotes the parents of node $i$ in $\mathcal{G}$ \cite{lauritzen96}.

Let $\mathcal{M}(\mathcal{G})$ denote the set of strictly positive densities that are Markov with respect to $\mathcal{G}$. Two DAGs $\mathcal{G}_1$ and $\mathcal{G}_2$ for which $\mathcal{M}(\mathcal{G}_1)=\mathcal{M}(\mathcal{G}_2)$ are said to be \emph{Markov equivalent} and belong to the same MEC \cite{andersson97}. \citet{verma90} gave a graphical criterion for Markov equivalence: two DAGs $\mathcal{G}_1$ and $\mathcal{G}_2$ belong to the same MEC if and only if they have the same skeleta (i.e., underlying undirected graph) and v-structures (i.e., induced subgraphs $i\to j\leftarrow k$).

Under perfect interventions, the identifiability of $\mathcal{G}$ improves from its MEC to its perfect-$\mathcal{I}$-MEC, which has the following graphical characterization \cite{hauser12}.

\begin{theorem} \label{the:perfect-i-mec}
Let $\mathcal{I} \subset \frak{P}([p])$\footnote{power set of $[p]$} be a conservative%
 ~(multi)-set of intervention targets, i.e. $\forall j \in [p], ~\exists I \in \mathcal{I}$ s.t. $j \notin I$. Two DAGs $\mathcal{G}$, $\mathcal{H}$ belong to the same perfect-$\mathcal{I}$-MEC if and only if $\mathcal{G}_{(I)}$, $\mathcal{H}_{(I)}$ are in the same MEC for all $I \in \mathcal{I}$,~where $\mathcal{G}_{(I)}$ denotes the sub-DAG of $\mathcal{G}$ with vertex set $[p]$ and edge set $\{(a \rightarrow b) | (a \rightarrow b) \in E, b \notin I \}$~and~similarly~for~$\mathcal{H}_{(I)}$. 
\end{theorem}
In this work, we extend this result to general interventions.
\begin{definition} \label{def:soft-int}
Under a (general) \emph{intervention} on target $I \subset [p]$, the \emph{interventional distribution} $f^{(I)}$ can be factorized as%
\begin{equation}\label{eq:int-factor}
f^{(I)}(X) = \prod_{i \in I} f^{(I)}(X_i|X_{\pa_{\mathcal{G}}(i)})\prod_{j \notin I} f^{(\emptyset)}(X_j|X_{\pa_{\mathcal{G}}(j)})
\end{equation}%
where $f^{(I)}$ and $f^{(\emptyset)}$ denote the interventional and observational distributions over $X$ respectively. Note that $f^{(I)}(X_j|X_{\pa_{\mathcal{G}}(i)})=f^{(\emptyset)}(X_j|X_{\pa_{\mathcal{G}}(i)}), ~\forall j \notin I$, i.e. the conditional distributions of non-targeted variables are {\it invariant} to the intervention.
\end{definition}

\subsection{Main Results} \label{sec:main_iden}

Let $\{f^{(I)}\}_{I \in \mathcal{I}}$ denote a collection of distributions over $X$ indexed by $I \in \mathcal{I}$.

\begin{definition}\label{def:M}
For a DAG $\mathcal{G}$ and interventional target set $\mathcal{I}$, define%
\[\begin{split}
&\mathcal{M}_{\mathcal{I}}(\mathcal{G}):= \{~\{f^{(I)}\}_{I \in \mathcal{I}} \mid
\forall I,J \in \mathcal{I}: f^{(I)} \in \mathcal{M}(\mathcal{G}) ~\text{and} ~ \\
&f^{(I)}(X_j|X_{\pa_{\mathcal{G}}(j)})=f^{(J)}(X_j|X_{\pa_{\mathcal{G}}(j)}), \forall j \notin I \cup J\}
\end{split}\]
\end{definition}
$\mathcal{M}_{\mathcal{I}}(\mathcal{G})$
 contains exactly the sets of interventional distributions (Definition \ref{def:soft-int}) that can be generated from a causal model with DAG $\mathcal{G}$ by intervening on $\mathcal{I}$
(see Supplementary Material for details). 
We therefore use $\mathcal{M}_{\mathcal{I}}(\mathcal{G})$ to formally define equivalence classes of DAGs under interventions.

\begin{definition}[$\mathcal{I}$-Markov Equivalence Class]
Two DAGs $\mathcal{G}_1$ and $\mathcal{G}_2$ for which $\mathcal{M}_{\mathcal{I}}(\mathcal{G}_1) = \mathcal{M}_{\mathcal{I}}(\mathcal{G}_2)$ belong to the same \emph{$\mathcal{I}$-Markov equivalence class ($\mathcal{I}$-MEC)}.
\end{definition}
 
From here, we extend the Markov property to the interventional setting to establish a graphical criterion for $\mathcal{I}$-MECs. We start by introducing the following graphical framework for representing DAGs under interventions.

\begin{figure}[t]
	\center
	\includegraphics[scale=0.6]{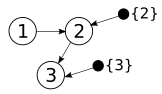}
	\caption{Let $\mathcal{G}$ be the DAG $1 \rightarrow 2 \rightarrow 3$ and let $\mathcal{I} = \{\emptyset, \{2\}, \{3\} \}$. The interventional DAG $\mathcal{G}^{\mathcal{I}}$ is shown above. Solid circles represent the $\mathcal{I}$-vertices, which are parameters indicating the intervention, and open circles represent random variables.}
	 \label{fig:aug-dag}
\end{figure}

\begin{definition}\label{def:i-dag}
Let $\mathcal{G} = ([p], E)$ be a DAG and let $\mathcal{I}$ be a collection of intervention targets. 
The \emph{interventional DAG}%
\footnote{In some previous work, interventions have been treated as additional variables of the causal system, which at first glance results in a DAG similar to the $\mathcal{I}$-DAG. The challenge then is that the new variables are deterministically related to each other, which leads to faithfulness violations (see \citet{magliacane16}). We have avoided this problem by treating the interventions as parameters instead of variables.}%
~(\emph{$\mathcal{I}$-DAG})  $\mathcal{G}^{\mathcal{I}}$ is the graph $\mathcal{G}$ augmented with $\mathcal{I}$-vertices $\{\zeta_I\}_{I \in \mathcal{I}, I \neq \emptyset}$ and $\mathcal{I}$-edges $\{\zeta_I \rightarrow i \}_{i \in I \in \mathcal{I}, I \neq \emptyset}$.
\end{definition}

Figure~\ref{fig:aug-dag} gives a concrete example of an $\mathcal{I}$-DAG. Note that each $\mathcal{I}$-vertex represents an intervention, and an $\mathcal{I}$-edge from an $\mathcal{I}$-vertex to a regular node $i$ indicates that $i$ is targeted under that intervention. Next, we define the $\mathcal{I}$-Markov property for $\mathcal{I}$-DAGs, analogous to the Markov property based on d-separation for DAGs. For now, we make the simplifying assumption that $\emptyset \in \mathcal{I}$; in Section~\ref{sec_no_obs_data}, we will show that this assumption can be made without loss of generality.

\begin{definition} [$\mathcal{I}$-Markov Property%
]\label{def:i-markov} 
	Let $\mathcal{I}$ be a set of intervention targets such that $\emptyset \in \mathcal{I}$, and suppose $\{f^{(I)}\}_{I \in \mathcal{I}}$ is a set of (strictly positive) probability distributions over $X_1, \cdots, X_p$ indexed by $I \in \mathcal{I}$. $\{f^{(I)}\}_{I \in \mathcal{I}}$ satisfies the \emph{$\mathcal{I}$-Markov property} with respect to the $\mathcal{I}$-DAG $\mathcal{G}^{\mathcal{I}}$ iff
\begin{enumerate}
\item $X_A \indep X_B\mid X_C$ for any $I \in \mathcal{I}$ and any disjoint $A, B, C \subset [p]$ such that $C$ d-separates $A$ and $B$~in~$\mathcal{G}$.
\item $f^{(I)}(X_A | X_C)=f^{(\emptyset)}(X_A | X_C)$ for any $I\in\mathcal{I}$ and any disjoint $A, C \subset [p]$ such that  $C \cup \zeta_{\mathcal{I}\backslash I}$ d-separates $A$ and $\zeta_I$ in $\mathcal{G}^{\mathcal{I}}$, where $\zeta_\emptyset:=\emptyset$ and $\zeta_{\mathcal{I} \backslash I} := \{\zeta_J ~|~ J \in \mathcal{I}, J \neq I\}$.
\end{enumerate}
\end{definition}

The first condition is simply the Markov property for DAGs based on d-separation. The second condition generalizes this property to $\mathcal{I}$-DAGs by relating d-separation between $\mathcal{I}$-vertices and regular vertices to the \emph{invariance} of conditional distributions across interventions. We note that the $\mathcal{I}$-Markov property is very similar to the ``missing-link compatibility" by \citet{bare11}

\begin{example} \label{ex:i-markov}
Consider again the augmented graph $\mathcal{G}^{\mathcal{I}}$ from Figure \ref{fig:aug-dag}, and suppose $\{f^{(I)} \}_{I \in \mathcal{I}}$ satisfies the $\mathcal{I}$-Markov property with respect to $\mathcal{G}^{\mathcal{I}}$. Then $\{f^{(I)} \}_{I \in \mathcal{I}}$ satisfies the following invariance relations based on d-separation: (1) $f^{(\emptyset)}(X_1) = f^{(\{2\})}(X_1) = f^{(\{3\})}(X_1)$;
(2) $f^{(\emptyset)}(X_3|X_2)=f^{(\{2\})}(X_3|X_2)$; (3) $f^{(\emptyset)}(X_2|X_1)=f^{(\{3\})}(X_2|X_1)$.

\end{example}

Having defined the $\mathcal{I}$-Markov property, we now formalize its relationship to $\mathcal{I}$-MECs.

\begin{proposition} \label{lem:i-markov}
Suppose $\emptyset \in \mathcal{I}$. Then $\{f^{(I)}\}_{I \in \mathcal{I}} \in \mathcal{M}_{\mathcal{I}}(\mathcal{G})$ if and only if $\{f^{(I)}\}_{I \in \mathcal{I}}$ satisfies the $\mathcal{I}$-Markov property with respect to $\mathcal{G}^{\mathcal{I}}$.
\end{proposition}
This result states that DAGs are in the same $\mathcal{I}$-MEC if and only if the d-separation statements of their $\mathcal{I}$-DAGs imply the same conditional invariances and independences based on the $\mathcal{I}$-Markov property. We now state the main result of this section: the graphical characterization of $\mathcal{I}$-MECs.

\begin{theorem}\label{the:i-mec}
Suppose $\emptyset \in \mathcal{I}$. Two DAGs $\mathcal{G}_1$ and $\mathcal{G}_2$ belong to the same $\mathcal{I}$-MEC if and only if their $\mathcal{I}$-DAGs $\mathcal{G}_1^{\mathcal{I}}$ and $\mathcal{G}_2^{\mathcal{I}}$ have the same skeleta and v-structures.
\end{theorem}

The proof of this theorem uses the following weak completeness result for the $\mathcal{I}$-Markov property.

\begin{lemma} \label{lem:i-markov-comp}
For any disjoint $A,C \subset [p]$ and any $J \in \mathcal{I}$ such that $C \cup \zeta_{\mathcal{I}\backslash J}$ does not d-separate $A$ and $\zeta_J$ in $\mathcal{G}^{\mathcal{I}}$, there exists some $\{f^{(I)}\}_{I \in \mathcal{I}}$ that satisfies the $\mathcal{I}$-Markov property with respect to $\mathcal{G}^{\mathcal{I}}$ with $f^{(\emptyset)}(X_A|X_C) \neq f^{(J)}(X_A|X_C)$.
\end{lemma}

\begin{proof}[Proof of Theorem \ref{the:i-mec}]
If $\mathcal{G}^{\mathcal{I}}_1$ and $\mathcal{G}^{\mathcal{I}}_2$ have the same skeleta and v-structures, then they satisfy the same d-separation statements, and hence $\mathcal{M}_{\mathcal{I}}(\mathcal{G}_1) = \mathcal{M}_{\mathcal{I}}(\mathcal{G}_2)$ by Proposition \ref{lem:i-markov}. If $\mathcal{G}^{\mathcal{I}}_1$ and $\mathcal{G}^{\mathcal{I}}_2$ do not have the same skeleta or v-structures, then (a) $\mathcal{G}_1$ and $\mathcal{G}_2$ do not have the same skeleta or v-structures, or (b) there exists $I \in \mathcal{I}$ and $j \in [p]$ such that $\zeta_I \rightarrow j$ is part of a v-structure in one $\mathcal{I}$-DAG and not the other. In case (a), $\mathcal{G}_1$ and $\mathcal{G}_2$ do not belong to the same MEC \cite{verma90}, so they also cannot belong to the same $\mathcal{I}$-MEC by the first condition in Definition~\ref{def:i-markov}. 
In case (b), suppose without loss of generality that $\zeta_I \rightarrow j$ is part of a v-structure in $\mathcal{G}_1^{\mathcal{I}}$ but not in $\mathcal{G}_2^{\mathcal{I}}$ for some $I \in \mathcal{I}$ and some $j \in [p]$. Then $j$ has a neighbor $k\in[p]\setminus\{j\}$ with orientation $k \rightarrow j$ in $\mathcal{G}_1^{\mathcal{I}}$ and $j \rightarrow k$ in $\mathcal{G}_2^{\mathcal{I}}$. Thus, $k$ and $\zeta_I$ are d-connected in $\mathcal{G}_1^{\mathcal{I}}$ given $pa_{\mathcal{G}_2^{\mathcal{I}}}(k)$ but d-separated in $\mathcal{G}_2^{\mathcal{I}}$ given $pa_{\mathcal{G}_2^{\mathcal{I}}}(k)$. Hence by Lemma~\ref{lem:i-markov-comp}, there exists some $\{f^{(I)}\}_{I \in \mathcal{I}}$ that satisfies the $\mathcal{I}$-Markov property with respect to $\mathcal{G}_1^{\mathcal{I}}$ but not $\mathcal{G}_2^{\mathcal{I}}$ and thus $\mathcal{M}_\mathcal{I}(\mathcal{G}_1) \neq \mathcal{M}_{\mathcal{I}}(\mathcal{G}_2)$.
\end{proof}

\begin{figure}
\subfigure[]{
\includegraphics[scale=0.6]{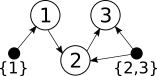}
}
\subfigure[]{
\includegraphics[scale=0.6]{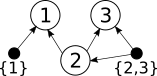}
}
\subfigure[]{
\includegraphics[scale=0.6]{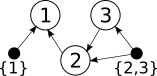}
}
\caption{Example of $\mathcal{I}$-DAGs for 3-node graphs with $\mathcal{I} = \{\emptyset, \{1\}, \{2,3\} \}$}
\label{fig:imec}
\end{figure}

\begin{example}
The three DAGs in Figure \ref{fig:imec} belong to the same MEC. Given interventions on $\mathcal{I} = \{\emptyset, \{1\}, \{2,3\} \}$, by Theorem \ref{the:i-mec}, DAG (a) is not in the same $\mathcal{I}$-MEC as DAGs (b-c) due to its lack of v-structure $\zeta_{\{1\}} \rightarrow 1 \leftarrow 2$. The intervention improves the identifiability of these structures.
\end{example}

It is straightforward to show that our graphical criterion of $\mathcal{I}$-MECs when $\emptyset \in \mathcal{I}$ is equivalent to the characterization of perfect-$\mathcal{I}$-MECs by \citet{hauser12} for perfect interventions, which proves their conjecture.

\begin{corollary} \label{cor:perfect-i-mec}
When $\emptyset \in \mathcal{I}$, two DAGs $\mathcal{G}_1$ and $\mathcal{G}_2$ are in the same $\mathcal{I}$-MEC iff they are in the same perfect-$\mathcal{I}$-MEC.
\end{corollary}

\subsection{Extension to $\emptyset \notin \mathcal{I}$}
\label{sec_no_obs_data}
The identifiability results for perfect-$\mathcal{I}$-MECs by \citet{hauser12} hold for conservative $\mathcal{I}$, while our results for $\mathcal{I}$-MECs requires a stronger assumption, namely that $\emptyset \in \mathcal{I}$ (i.e. observational data is available). While this assumption is not restrictive in practice, it raises the question of whether our results can be extended to conservative sets of targets when $\emptyset\notin\mathcal{I}$. The following example shows that our current graphical characterization of $\mathcal{I}$-MECs (Theorem \ref{the:i-mec}) does not generally hold under this weaker assumption.

\begin{example}
Let $\mathcal{G}$ be the causal DAG $1 \rightarrow 2$ and let $\mathcal{I}=\{\{1\}, \{2\} \}.$ 
The interventional distributions have the factorization $f^{(\{1\})}(X) = f^{(\{1\})}(X_1)f^{(\emptyset)}(X_2|X_1)$ and $f^{(\{2\})}(X) = f^{(\emptyset)}(X_1)f^{(\{2\})}(X_2|X_1)$ respectively, according to Definition \ref{def:soft-int}. Any distributions with this factorization can also be written as $f^{(\{1\})}(X) = g^{(\emptyset)}(X_2)g^{(\{1\})}(X_1|X_2)$ and $f^{(\{2\})}(X) = g^{(\{2\})}(X_2)g^{(\emptyset)}(X_1|X_2)$ for an appropriate choice of $g^{(\emptyset)}$, $g^{(\{1\})}$ and $g^{(\{2\})}$. Thus, $\mathcal{G}_1$ and $\mathcal{G}_2$ belong to the same $\mathcal{I}$-MEC (i.e., $\mathcal{M}_{\mathcal{I}}(\mathcal{G}_1)=\mathcal{M}_{\mathcal{I}}(\mathcal{G}_2)$). But $\mathcal{G}_1^{\mathcal{I}}$ and $\mathcal{G}_2^{\mathcal{I}}$ do not have the same v-structures, contradicting the graphical criterion of Theorem \ref{the:i-mec}.
\end{example}

The following theorem extends our graphical characterization of $\mathcal{I}$-MECs to conservative sets of intervention targets when we don't necessarily have $\emptyset\in\mathcal{I}$. The proof of this result is provided in the Supplementary Material.

\begin{theorem}
\label{th:imec-2}
Let $\mathcal{I} \subset \frak{P}([p])$ be a conservative set of intervention targets. Two causal DAGs $\mathcal{G}_1$ and $\mathcal{G}_2$ belong to the same $\mathcal{I}$-MEC if and only if for all $I\in \mathcal{I}$ the interventional DAGs $\mathcal{G}_1^{\tilde{\mathcal{I}}_I}$ and $\mathcal{G}_2^{\tilde{\mathcal{I}}_I}$ have the same skeletons and v-structures, where 
\[
\tilde{\mathcal{I}}_I := \{\emptyset, \{I \cup J\}_{J \in \mathcal{I}, J \neq I} \}
\]
\end{theorem}

The proof formalizes the following intuition: in the absence of an observational dataset, we can relabel one of the interventional datasets (i.e. from intervening on $I$) as the observational one; or equivalently, we ``pretend" that our datasets are obtained under interventions on $\tilde{\mathcal{I}}_I$ instead of $\mathcal{I}$. Then two DAGs cannot be distinguished under interventions on $\mathcal{I}$ if and only if this also holds for $\tilde{\mathcal{I}}_I$, for all $I \in \mathcal{I}$. Note that if $\emptyset \in \mathcal{I}$, then this statement is equivalent to Theorem~\ref{the:i-mec}. 
Hence the assumption $\emptyset \in\mathcal{I}$ in Section~\ref{sec:main_iden} can be made without loss of generality and our identifiability results extend to all conservative sets of intervention targets.

\begin{figure}
\centering
\subfigure[]{
\includegraphics[scale=0.6]{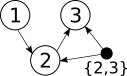}
}
\subfigure[]{
\includegraphics[scale=0.6]{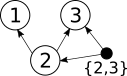}
}
\subfigure[]{
\includegraphics[scale=0.6]{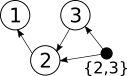}
}
\caption{Example of $\tilde{\mathcal{I}}_{\{2\}}$-DAGs for 3-node graphs with $\mathcal{I} = \{\{2\}, \{3\} \}$. Note that the $\tilde{\mathcal{I}}_{\{3\}}$-DAGs are identical since $\tilde{\mathcal{I}}_{\{2\}} = \tilde{\mathcal{I}}_{\{3\}} = \{\emptyset, \{2,3\} \}$ in this case.}
\label{fig:imec-2}
\end{figure}

\begin{example}
The three DAGs in Figure \ref{fig:imec-2} belong to the same MEC. Given interventions on $\mathcal{I} = \{\{2\}, \{3\} \}$, by Theorem \ref{th:imec-2}, DAG (a) is not in the same $\mathcal{I}$-MEC as DAGs (b-c) due to its v-structure $\zeta_{\{2,3\}} \rightarrow 2 \leftarrow 1$. The intervention improves the identifiability of these structures.
\end{example}

\section{Consistent algorithm for learning $\mathcal{I}$-MECs}
\label{sec_alg}
Having shown that the $\mathcal{I}$-MEC of a causal DAG can be identified from general interventions, we now propose a permutation-based algorithm for learning the $\mathcal{I}$-MEC. The algorithm takes interventional datasets obtained under \emph{general interventions} with \emph{known targets} $\mathcal{I}$ and returns a DAG in the correct $\mathcal{I}$-MEC. 

\subsection{Preliminaries}

\begin{figure}[t]
	\center
	\includegraphics[scale=0.45]{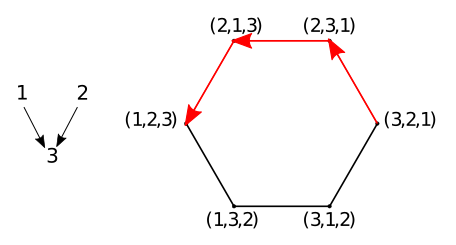}
	\vspace{-5mm}
	\caption{Left: DAG corresponding to permutations (1,2,3) or (2,1,3). Right: Illustration of greedy search over the space of permutations for $p=3$, starting at (3,2,1). The space of permutations is represented by a polytope known as the \emph{permutahedron} in which each node corresponds to a permutation and edges connect neighboring transpositions. A greedy search corresponds to a greedy edge walk (red arrows) over the permutahedron.}
	 \label{fig:permwalk}
\end{figure}

Permutation-based causal inference algorithms search for a permutation $\pi^*$ that is consistent with the topological order of the true causal DAG $\mathcal{G^*}$, i.e. if $(i,j)$ is an edge in $\mathcal{G}^*$ then $i<j$ in $\pi^*$ (Figure \ref{fig:permwalk}, left). Given $\pi^*$, $\mathcal{G}^*$ can then be determined by learning an undirected graph over the nodes and orienting the edges according to the order $\pi^*$.

To find $\pi^*$, one option is to do a greedy search over the space of permutations by tranposing neighboring nodes and optimizing a score function (Figure \ref{fig:permwalk}, right). In \citet{solus17}, the authors propose an algorithm called \emph{Greedy Sparsest Permutations (GSP)} that uses a score function based on CI relations. Specifically, the score of a given permutation $\pi$ is the number of edges in its \emph{minimal I-map} $\mathcal{G}_{\pi} = ([p], E_{\pi})$, which is the sparsest DAG consistent with $\pi$ such that $f^{(\emptyset)}$ is Markov with respect to $\mathcal{G}_{\pi}$. Since the score is only guaranteed to be weakly decreasing on any path from $\pi$ to $\pi^*$, the algorithm iteratively uses a depth-first-search. Additionally, instead of considering all neighboring transpositions of $\pi$ in the search, GSP only transposes neighboring nodes in the permutation that are connected by \emph{covered} edges\footnote{An edge $(i,j)$ in a DAG $\mathcal{G}$ is \emph{covered} if $\pa_{\mathcal{G}}(i)=\pa_{\mathcal{G}}(j)\!\setminus\!\{i\}$.} in $\mathcal{G}_{\pi}$, which improves the efficiency of the algorithm. Under the assumptions of causal sufficiency and faithfulness%
\footnote{\emph{Causal sufficiency} is the assumption that there are no hidden latent confounders, and \emph{faithfulness} implies that all CI relations of the observational distribution $f^{\emptyset}$ are implied by d-separation in $\mathcal{G}$.}%
, GSP is \emph{consistent} in that it returns a permutation $\tau$ where $\mathcal{G}_{\tau}$ is in the same MEC as the true DAG $\mathcal{G}^*$ \cite{solus17, Mohammadi18}. However, GSP does not use data from interventions, so it is not guaranteed to return a DAG in the correct $\mathcal{I}$-MEC.

Perfect-IGSP extends GSP to incorporate data from interventions \cite{wang17}. 
However, the consistency result of perfect-IGSP requires the interventional data to come from perfect interventions. 
This motivates our development of a new algorithm, \emph{IGSP (or general-IGSP)}, which is provably consistent for finding the $\mathcal{I}$-MEC of $\mathcal{G}^*$ when the data come from general interventions.

\subsection{Main Results}
In Algorithm \ref{alg:soft_igsp}, we present \emph{IGSP}, a greedy permutation-based algorithm for recovering the $\mathcal{I}$-MEC of $\mathcal{G}^*$ from $\{f^{(I)}\}_{I \in \mathcal{I}}$ for general interventions with known targets $\mathcal{I}$. 

Similar to GSP, IGSP starts with a permutation $\pi$ and implements depth-first-search to look for a permutation $\tau$ such that $|\mathcal{G}_{\tau}|<|\mathcal{G}_{\pi}|$, where $\mathcal{G}_{\tau}$ and $\mathcal{G}_{\pi}$ are the minimal I-maps of $\tau$ and $\pi$ respectively; and iterates until no such permutation can be found. One difference from GSP is that in each step of the search,
IGSP only transposes neighboring nodes that are connected by \emph{$\mathcal{I}$-covered} edges%
\footnote{Correction from the previous version presented at ICML 2018.}%
 in the corresponding minimal I-map.

\begin{definition}
 A covered edge $i \rightarrow j$ in a DAG $\mathcal{G}$ is \emph{$\mathcal{I}$-covered} if $f^{(\{i\})}(X_j)=f^{(\emptyset)}(X_j)$ when $\{i\} \in \mathcal{I}$.
\end{definition}

The use of $\mathcal{I}$-covered edges restricts the search space and ensures that we do not consider permutations that contradict order relations derived from the intervention experiments. Furthermore, the transposition of neighboring nodes connected by $\mathcal{I}$-covered edges that are also \emph{$\mathcal{I}$-contradictory} edges is prioritized during the search.

\begin{definition}
	\label{def:i-cont-edge}
Let $\neigh_{\mathcal{G}}(i)$ denote the neighbors of node $i$ in a DAG $\mathcal{G}$. An edge $i \rightarrow j$ in $\mathcal{G}$ is \emph{$\mathcal{I}$-contradictory} if at least one of the following two conditions hold: 

(1) There exists a set $S \subset \neigh_{\mathcal{G}}(j) \backslash \{i\}$ such that $f^{(\emptyset)}(X_j|X_S)=f^{(I)}(X_j|X_S)$ for all $I \in \mathcal{I}_{i \backslash j}$ ;

(2) $f^{(\emptyset)}(X_i|X_S) \neq f^{(I)}(X_i|X_S)$ for some $I\in\mathcal{I}_{j \backslash i}$, for all $S \subset \neigh_{\mathcal{G}}(i) \backslash \{j\}$.
\end{definition}

$\mathcal{I}$-contradictory edges are prioritized because they violate the $\mathcal{I}$-Markov property (Definition \ref{def:i-markov}). Thus, a DAG in the correct $\mathcal{I}$-MEC should minimize the number of $\mathcal{I}$-contradictory edges. 
Evaluating whether edges are $\mathcal{I}$-contradictory requires invariance tests that grow with the maximum degree of $\mathcal{G}_{\pi}$. When $\mathcal{I}$ consists of only single-node interventions, a modified definition of $\mathcal{I}$-contradictory edges can be used to reduce the number of tests.

\begin{definition}
Let $\mathcal{I}$ be a set of intervention targets such that $\{i\} \in \mathcal{I}$ or $\{j\} \in \mathcal{I}$. The edge $i \rightarrow j$ is \emph{$\mathcal{I}$-contradictory} if either of the following is true: 

(1) $\{i\} \in \mathcal{I}$ and $f^{\{i\}}(X_j)=f^{\emptyset}(X_j)$; or

(2) $\{j\} \in \mathcal{I}$ and $f^{\{j\}}(X_i)\neq f^{\emptyset}(X_i)$.
\end{definition}	

In the special case where we only have single-node interventions, the number of invariance tests no longer depends on the maximum degree of $\mathcal{G}_{\pi}$ under this simplification.
	
Unlike perfect-IGSP, which is consistent only under perfect interventions, our method is consistent for general interventions under the following two assumptions:

\begin{assumption} \label{ass:1}
Let $I \in \mathcal{I}$ with $i\in I$. Then $f^{(I)}(X_j)\neq f^{(\emptyset)}(X_j)$ for all descendants $j$ of $i$.
\end{assumption}

\begin{assumption} \label{ass:2}
Let $I \in \mathcal{I}$ with $i\in I$. Then $f^{(I)}(X_j|X_S)\neq f^{(\emptyset)}(X_j|X_S)$ for any child $j$ of $i$ such that $j\notin I$ and  for all $S \subset \neigh_{\mathcal{G}^*}(j) \setminus\{i\}$, where $\neigh_{\mathcal{G}^*}(j)$ denotes the neighbors of node $j$ in $\mathcal{G}^*$.
\end{assumption}

Both assumptions are strictly weaker than the faithfulness assumption on the $\mathcal{I}$-DAG. Assumption \ref{ass:1} extends the assumption by \citet{tian01} to interventions on multiple nodes. It essentially requires interventions on upstream nodes to affect downstream nodes. Assumption \ref{ass:2} is similarly intuitive and requires the distribution of $X_j$ to change under an intervention on its parent $X_i$ as long as $X_i$ is not part of the conditioning set. 

The main result of this section is the following theorem, which states the consistency of IGSP.
\begin{theorem} \label{the:consistency}
Algorithm \ref{alg:soft_igsp} is consistent under assumptions \ref{ass:1} and \ref{ass:2}, faithfulness of $f^{(\emptyset)}$ with respect to $\mathcal{G}$, and causal sufficiency. When $\mathcal{I}$ only contains single-variable interventions, assumption \ref{ass:2} is not required for the correctness of the algorithm.
\end{theorem}

\begin{algorithm}[tb]
   \caption{IGSP for general interventions}
   \label{alg:soft_igsp}
\begin{algorithmic}
   \STATE {\bfseries Input:} A collection of intervention targets $\mathcal{I}$ with $\emptyset\in\mathcal{I}$, samples from distributions $\{{f}^{(I)}\}_{I \in \mathcal{I}}$, and a starting permutation~$\pi_0$.
   \STATE {\bfseries Output:} A permutation $\tau$ and associated I-map $\mathcal{G}_\tau$
   \STATE Set $\pi = \pi_0$, $\mathcal{G_{\pi}} :=$ minimal I-map of $\pi$.
   \REPEAT 
   \STATE Using a depth-first-search with root $\pi$, search for a permutation $\tau$ with minimal I-map $\mathcal{G}_\tau$ such that $\vert \mathcal{G}_{\pi} \vert > \vert \mathcal{G}_\tau \vert$ that is connected to $\mathcal{G}_{\pi}$ by a sequence of $\mathcal{I}$-covered edge reversals, with priority given to $\mathcal{I}$-contradictory edge reversals. If $\tau$ exists, set $\pi = \tau$, $\mathcal{G_{\pi}} =\mathcal{G_{\tau}}$.
   \UNTIL{No such $\tau$ can be found.}
   \STATE Return the permutation $\tau$ and the associated I-map $\mathcal{G}_\tau$ with $\vert \mathcal{G}_{\tau} \vert = \vert \mathcal{G}_{\pi} \vert$ that minimizes the number of $\mathcal{I}$-contradicting edges.
   
\end{algorithmic}
\end{algorithm}

\subsection{Implementation of Algorithm~\ref{alg:soft_igsp}} \label{sec:pool}

\textbf{Testing for invariance:} To test whether a (conditional) distribution $f^{(I)}(X_i | X_S)$ is invariant, we used a method proposed by \citet{heinze17} that we found to work well in practice. Briefly, we test whether $X_i$ is independent of the \emph{index} of the interventional dataset given $X_S$, using the 
HSIC gamma test \cite{gretton2005kernel}.

\textbf{Data pooling for CI testing:} 
Let $\an_{\mathcal{G}_\pi}(i)$ denote the ancestors of node $i$ in $\mathcal{G}_\pi$. After reversing an $\mathcal{I}$-covered edge $(i,j)$, updating $\mathcal{G}_\pi$ requires testing if $X_i \indep X_k \mid X_{\an_{\mathcal{G}_\pi}(i) \setminus \{k\}}$ for $k \in \pa_{\mathcal{G}_\pi}(i)$ under the observational distribution $f^{(\emptyset)}$. By combining the interventional data with the observational data in a provably correct manner, we can increase the power of the CI tests, which is useful when the sample sizes are limited.
In the Supplementary Material, we present a proposition giving sufficient conditions under which CI relations hold when the data come from a mixture of interventional distributions, and use this to derive a set of checkable conditions on $\mathcal{G}_{\pi}$ for determining which datasets can be combined to test $X_i \indep X_k \mid X_{\an_{\mathcal{G}_\pi}(i) \setminus \{k\}}$ for $k \in \pa_{\mathcal{G}_\pi}(i)$.

\section{Empirical Results}
\label{sec_res}
\subsection{Experiments on simulated datasets}

\begin{figure*}[t] 
\center
\subfigure[]{
\includegraphics[scale=0.3]{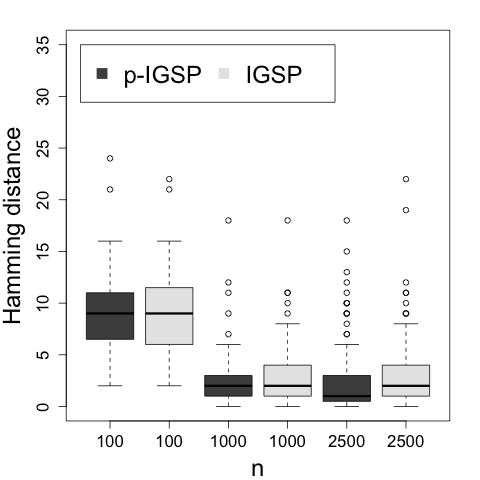}
}
\subfigure[]{
\includegraphics[scale=0.3]{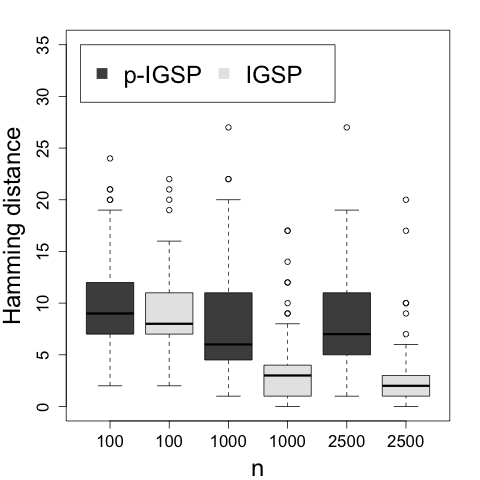}
}
\subfigure[]{
\includegraphics[scale=0.3]{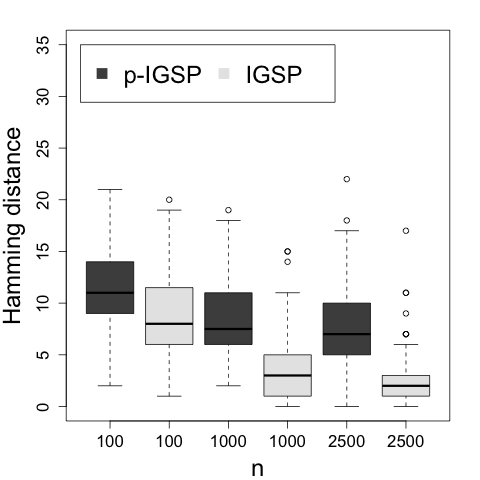}
}

\caption{Distributions of Hamming distances of recovered DAGs using IGSP and perfect-IGSP (p-IGSP) for 20-node graphs with single-variable (a) perfect, (b) imperfect, and (c) inhibitory interventions}
\label{fig:sim_igsp}
\end{figure*}

{\bf IGSP vs perfect-IGSP:} We compared Algorithm~\ref{alg:soft_igsp} to perfect-IGSP on the task of recovering the correct $\mathcal{I}$-MEC under three types of interventions: perfect, \emph{inhibiting}, and \emph{imperfect}. By an \emph{inhibiting} intervention, we mean an intervention that reduces the effect of the parents of the target node. This simulates a biological intervention such as a small-molecule inhibitor with a modest effect. By an \emph{imperfect} intervention, we mean an intervention that is perfect with probability $\alpha$ and ineffective with probability $1-\alpha$ for some $\alpha \in (0,1)$. This simulates biological experiments such as gene deletions that might not work in all cells. 

For each simulation, we sampled $100$ DAGs from an Erd\"os-Renyi random graph model with an average neighborhood size of $1.5$ and $p \in \{10, 20\}$ nodes. The data for each causal DAG $\mathcal{G}$ was generated using a linear structural equation model with independent Gaussian noise: $X = AX + \epsilon$, where $A$ is an upper-triangular matrix with edge weights $A_{ij} \neq 0$ if and only if $i \rightarrow j$, and $\epsilon \sim \mathcal{N}(0,\textrm{I}_p)$. For $A_{ij} \neq 0$, the edge weights were sampled uniformly from $[-1, -0.25] \cup [0.25, 1]$. We simulated perfect interventions on $i$ by setting the column $A_{,i} = 0$; inhibiting interventions by decreasing $A_{,i}$ by a factor of $10$; and imperfect interventions with a success rate of $\alpha=0.5$. Interventions were performed on all single-variable targets or all pairs of multiple-variable targets to maximally illuminate the difference between IGSP and perfect-IGSP.

Figure \ref{fig:sim_igsp} shows that IGSP outperforms perfect-IGSP on data from inhibiting and imperfect interventions and that the algorithms perform comparably on data from  perfect interventions (see also the Supplementary Material for further figures). 
These empirical comparisons corroborate our theoretical results that {IGSP is consistent for general types of interventions}, while perfect-IGSP is only consistent for perfect interventions. Consistency for general interventions is particularly important for applications to genomics, where it is usually not known a priori whether an intervention will be perfect; these results suggest we can use IGSP regardless of the type of intervention.

\begin{figure*}[t] 
\center
\subfigure[]{
\includegraphics[scale=0.3]{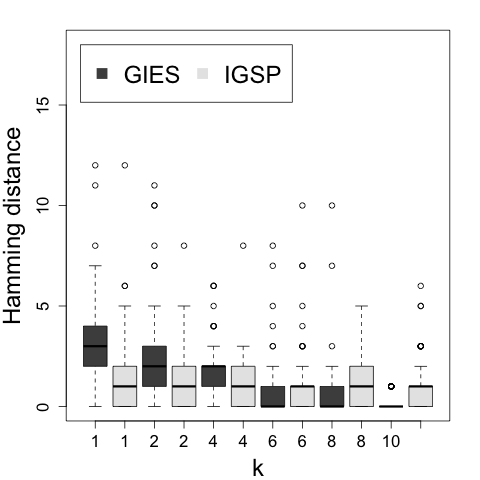}
}
\subfigure[]{
\includegraphics[scale=0.3]{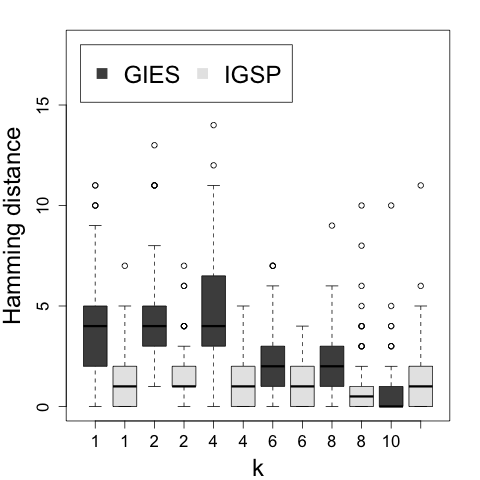}
}
\subfigure[]{
\includegraphics[scale=0.3]{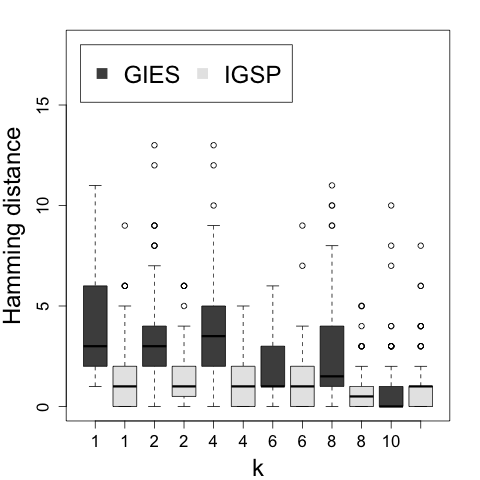}
}
\caption{Distributions of Hamming distances of recovered DAGs using IGSP and GIES for 10-node graphs with average edge density of 1.5 and single-node (a) perfect, (b) imperfect, and (c) inhibitory interventions on $k$ nodes}
\label{fig:sim_gies}
\end{figure*}

{\bf IGSP vs GIES:} GIES is an extension of the score-based causal inference algorithm, \emph{Greedy Equivalence Search (GES)}, to the interventional setting. Its score function incorporates the log-likelihood of the data based on the interventional distribution of Equation (\ref{eq:int-factor}), making it appropriate for learning DAGs under general interventions. Although GIES is not consistent in general \cite{wang17}, it has performed well in previous empirical studies \cite{hauser12,hauser15}. Additionally, both IGSP and GIES assume causal sufficiency and output DAGs, 
while the other methods mentioned in Section \ref{sec:related-work} do not output a DAG or use different assumptions. We therefore used GIES as a baseline for comparison. 

We evaluated IGSP and GIES on learning DAGs from different types of interventions, varying the number of interventional datasets ($|\mathcal{I}| = k \in \{1,2,4,6,8,10 \}$). The synthetic data was otherwise generated as described above. Figure \ref{fig:sim_gies} shows that IGSP in general significantly outperforms GIES. However, GIES performs better when the number of interventional datasets is large, i.e.~for $|\mathcal{I}| = 10$. This performance increase can be credited to the GIES score function which efficiently pools the interventional datasets.

\begin{figure*}[t]
	\label{fig:real_data}
	\subfigure[]{\label{fig:protein_skeleton}
		\includegraphics[scale=0.3]{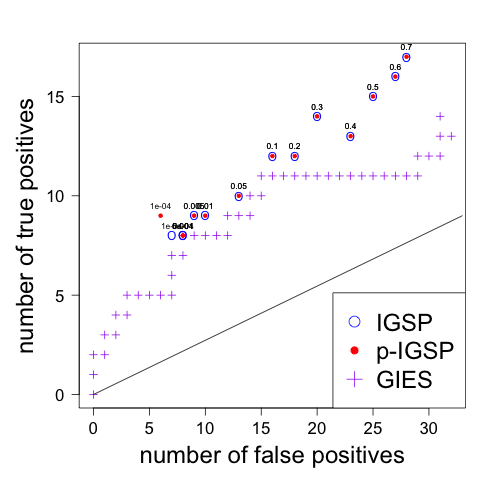}
	}
	\subfigure[]{\label{fig:protein_dag}
		\includegraphics[scale=0.3]{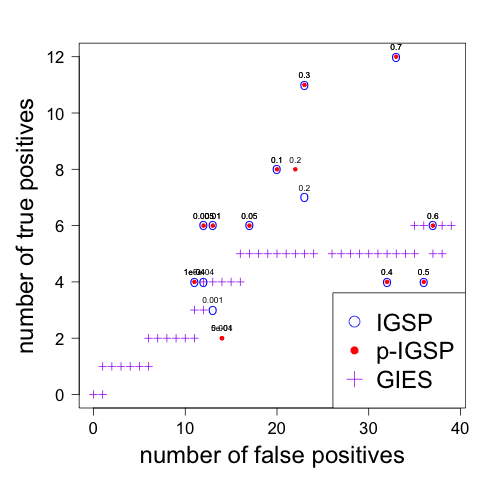}
	}
	\subfigure[]{\label{fig:perturb_seq}
		\includegraphics[scale=0.3]{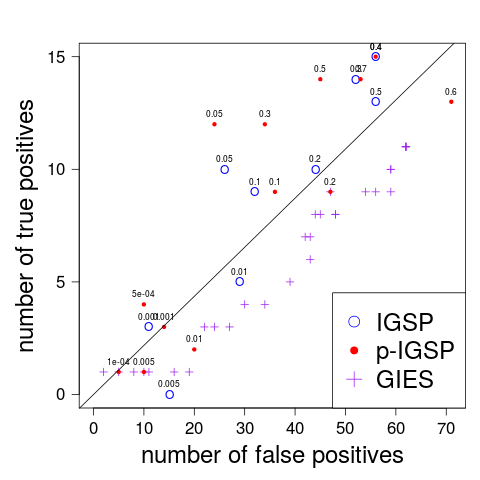}
	}\

	\caption{ROC plots evaluating IGSP, perfect-IGSP (p-IGSP) and GIES on learning the (a) skeleton and (b) DAG of the protein network from \citet{sachs05} and on (c) predicting the causal effects of interventions on a gene network from \cite{dixit16}}
\end{figure*}

\subsection{Experiments on Biological Datasets}
{\bf Protein Expression Dataset:} We evaluated our algorithm on the task of learning a protein network from a protein mass spectroscopy dataset \cite{sachs05}. The processed dataset consists of 5846 measurements of phosphoprotein and phospholipid levels from primary human immune system cells. Interventions on the network were perfect interventions corresponding to chemical reagents that strongly inhibit or activate certain signaling proteins. Figures \ref{fig:protein_skeleton} and \ref{fig:protein_dag} illustrate the ROC curves of IGSP, perfect-IGSP \cite{wang17} and GIES \cite{hauser15} on learning the skeleton and DAG of the ground-truth network respectively. 
We found that IGSP and perfect-IGSP performed comparably well on this dataset, which is consistent with our theoretical results. As expected, both IGSP and perfect-IGSP outperform GIES at recovering the true DAG,
since the former two algorithms have consistency guarantees in this regime while GIES does not.

{\bf Gene Expression Dataset:} We also evaluated IGSP on a single-cell gene expression dataset \cite{dixit16}. The processed dataset contains 992 observational and 13,435 interventional measurements of gene expression from bone marrow-derived dendritic cells. There are eight interventions in total, each corresponding to a targeted gene deletion using the CRISPR/Cas9 system. Since this dataset introduced the perturb-seq technique and was meant as a demonstration, we expected the interventions to be of high-quality and close to perfect.
We applied IGSP, perfect-IGSP, and GIES to learn causal DAGs over 24 transcription factors that modulate each other and play a critical role in regulating downstream genes. Since the ground-truth DAG is not available, we evaluated each learned DAG on its accuracy in predicting the effect of an intervention that was left out during inference, as described by \citet{wang17}. 
Figure \ref{fig:perturb_seq} shows that IGSP is competitive with perfect-IGSP, which suggests that the gene deletion interventions were close to perfect. Once again, both IGSP and perfect-IGSP outperform GIES on this dataset. 

\section{Discussion}

In this paper, we studied $\mathcal{I}$-MECs, the equivalence classes of causal DAGs that can be identified from a set of general (not necessarily perfect) intervention experiments. In particular, we provided a graphical characterization of $\mathcal{I}$-MECs and proved a conjecture of \citet{hauser12} showing that $\mathcal{I}$-MECs are equivalent to perfect-$\mathcal{I}$-MECs under basic assumptions. This result has important practical consequences, since it implies that general interventions provide similar causal information as perfect interventions despite being less invasive. An interesting problem for future research is to extend these identifiability results to the setting where the intervention targets are unknown. Such results would have wide-ranging implications, such as in genomics, where the interventions can have off-target effects. 

We also propose the first provably consistent algorithm, IGSP, for learning the $\mathcal{I}$-MEC from observational and general interventional data and apply it to protein and gene perturbation experiments. IGSP extends perfect-IGSP \cite{wang17}, which is only consistent for perfect interventions. 
In agreement with the theory, IGSP outperforms perfect-IGSP on data from non-perfect interventions and is competitive with perfect-IGSP on data from perfect interventions, 
thereby demonstrating the flexibility of IGSP to learn from different types of interventions. A challenge for future research is to scale algorithms like IGSP up to thousands of nodes, which would allow learning the entire gene network of a cell. The main bottleneck for scaling IGSP and an important area for future research is the development of accurate and fast conditional independence tests that can be applied under general distributional assumptions.

\pagebreak

\section*{Acknowledgements}
Karren D. Yang was partially supported by an NSF Graduate Fellowship.
Caroline Uhler was partially supported by NSF (DMS-1651995), ONR (N00014-17-1-2147), and a Sloan Fellowship.
\bibliographystyle{icml2018}
\bibliography{icml_2018}

\thispagestyle{empty}
\onecolumn
\appendix
\allowdisplaybreaks[1]

\section{Proofs from Section \ref{sec_ident}}
\subsection{Proofs from Section \ref{sec:main_iden}}

The following lemma formalizes the claim that $\mathcal{M}_{\mathcal{I}}(\mathcal{G})$ as given in Definition \ref{def:M} contains exactly the sets of interventional distributions that can be generated from a causal model with DAG $\mathcal{G}$ by intervening on $\mathcal{I}$.

\begin{lemma}\label{lem:M}
$\{f^{(I)}\}_{I \in \mathcal{I}} \in \mathcal{M}_{\mathcal{I}}(\mathcal{G})$ if and only if there exists $f^{(\emptyset)} \in \mathcal{M}(\mathcal{G})$ such that $\forall I \in \mathcal{I}, ~f^{(I)}$ factorizes according to Equation (\ref{eq:int-factor}) in Definition \ref{def:soft-int}.
\end{lemma}

\begin{proof}
Suppose there exists $f^{(\emptyset)} \in \mathcal{M}(\mathcal{G})$ such that $\forall I \in \mathcal{I}, ~f^{(I)}$ factorizes according to Equation (\ref{eq:int-factor}) in Definition \ref{def:soft-int}. Then $f^{(I)} \in \mathcal{M}(\mathcal{G})$ is trivially satisfied for all $I \in \mathcal{I}$. Also, we have $f^{(I)}(X_j|X_{\pa_{\mathcal{G}}(j)})=f^{(\emptyset)}(X_j|X_{\pa_{\mathcal{G}}(j)})~ \forall j \notin I$ and $I \in \mathcal{I}$. It follows that $f^{(\emptyset)}(X_j|X_{\pa_{\mathcal{G}}(j)})=f^{(I)}(X_j|X_{\pa_{\mathcal{G}}(j)})=f^{(J)}(X_j|X_{\pa_{\mathcal{G}}(j)}), ~ \forall j \notin I \cup J$ and all $I, J \in \mathcal{I}$. Therefore, $\{f^{(I)}\}_{I \in \mathcal{I}} \in \mathcal{M}_{\mathcal{I}}(\mathcal{G})$.

Conversely, suppose $\{f^{(I)}\}_{I \in \mathcal{I}} \in \mathcal{M}_{\mathcal{I}}(\mathcal{G})$. We will prove that there exists $f^{(\emptyset)} \in \mathcal{M}(\mathcal{G})$ such that $\forall I \in \mathcal{I}, ~f^{(I)}$ factorizes according to Equation (\ref{eq:int-factor}. Since $f^{(\emptyset)} \in \mathcal{M}(\mathcal{G})$, $f^{(\emptyset)}$ must factorize as
$f^{(\emptyset)}(X) =  \prod_{j \in [p]} f^{(\emptyset)}(X_j|X_{\pa_{\mathcal{G}}(j)})$
For each $j \in [p]$, let 
$f^{(\emptyset)}(X_j|X_{\pa_{\mathcal{G}}(j)}) = f^{(I_j)}(X_j|X_{\pa_{\mathcal{G}}(j)})$ for some $I_j \in \mathcal{I}$ s.t. $j \notin I$. 
If such a choice of $I_j$ does not exist, then let $f^{(\emptyset)}(X_j|X_{\pa_{\mathcal{G}}(j)})$ be an arbitrary strictly positive density.
Then note that for any $I \in \mathcal{I}$, we have
\[
\begin{split}
f^{(I)}(X) &= \prod_{j \in [p]} f^{(I)}(X_j|X_{\pa_{\mathcal{G}}(j)})\\
&= \prod_{i \in I} f^{(I)}(X_i|X_{\pa_{\mathcal{G}}(j)}) 
\prod_{j \notin I} f^{(I_j)}(X_j|X_{\pa_{\mathcal{G}}(j)})\\
&= \prod_{i \in I} f^{(I)}(X_i|X_{\pa_{\mathcal{G}}(j)}) 
\prod_{j \notin I} f^{(\emptyset)}(X_j|X_{\pa_{\mathcal{G}}(j)}),
\end{split}
\]
which completes the proof.
\end{proof}

\begin{proof} [Proof of Proposition \ref{lem:i-markov}]
To prove the ``if" direction, choose any $I \in \mathcal{I}$ and use the chain rule to factorize $f^{(I)}$ according to a topological ordering $\pi$ consistent with $\mathcal{G}$. Specifically, if we let $a_{\pi}(i)$ denote the nodes that precede $i$ in this ordering, then $f^{(I)}(X) = \prod_{i} f^{(I)}(X_i | X_{a_{\pi}(i)})$. Since every node is d-separated from its non-descendants given its parents, using condition (1) of the $\mathcal{I}$-Markov property, we can reduce the factorizations to $f^{(I)}(X) = \prod_{i} f^{(I)}(X_i | X_{pa(i)})$. Furthermore, since any node $i \notin I$ is d-separated from $I$ given its parents, using condition ($2$) of the $\mathcal{I}$-Markov property, we can substitute the interventional conditional distributions with the observational ones, resulting in $f^{I}(X) = \prod_{i \in I} f^{I}(X_i | X_{pa(i)}) \prod_{i \notin I} f^{\emptyset}(X_i | X_{pa(i)})$. Since this factorization holds for every $I \in \mathcal{I}$, $\{f^{I}\}_{I \in \mathcal{I}} \in \mathcal{M}_{\mathcal{I}}(\mathcal{G})$ by Lemma \ref{lem:M}.

To prove the ``only if" part of the statement, suppose $\{f^{I}\}_{I \in \mathcal{I}} \in \mathcal{M}_{\mathcal{I}}(\mathcal{G})$. By Lemma \ref{lem:M}, $f^{(I)}$ factorizes according to Equation (\ref{eq:int-factor}) and satisfies the Markov property with respect to $\mathcal{G}$ for all $I \in \mathcal{I}$. It follows that $f^{(I)}$ must also satisfy the Markov property based on d-separation with respect to $\mathcal{G}$ \cite{verma90}. Therefore, condition (1) of the $\mathcal{I}$-Markov property is satisfied.

To prove the second condition, choose any disjoint $A, C \subset [p]$ and any $I \in \mathcal{I}$, and suppose $C \cup \zeta_{\mathcal{I} \backslash \{I\}}$ d-separates $A$ from $\{\zeta_{I}\}$ in $\mathcal{G}^{\mathcal{I}}$. {Let $V_{An}$ be the ancestral set of $A$ and $C$ with respect to $\mathcal{G} = (V, E)$. Let $B' \subset V_{An}$ contain all nodes in $V_{An}$ that are d-connected to $\{\zeta_{I}\}$ in $\mathcal{G}^{\mathcal{I}}$ given $C \cup \zeta_{\mathcal{I} \backslash \{I\}}$, and let $A' = V_{An} \backslash (B' \cup C)$.} Since by Lemma \ref{lem:M}, $f^{(I)}$ factorizes over $\mathcal{G}$ according to Equation (\ref{eq:int-factor}) for every $I \in \mathcal{I}$, then choosing $\hat I \in \{\emptyset, I\}$ yields

\begin{align*}
f^{(\hat I)}(X) &= f^{(\hat I)}(X_{A'}, X_{B'}, X_C, X_{V\backslash V_{An}})\\~\\
&=\prod_{i \in A'} f^{(\hat I)}(X_i | X_{pa(i), \mathcal{G}})
\prod_{i \in C, pa_{\mathcal{G}}(i) \cap A' \neq \emptyset} f^{(\hat I)}(X_i | X_{pa(i), \mathcal{G}})\\
&\prod_{i \in C, pa_{\mathcal{G}}(i) \cap A' = \emptyset} f^{(\hat I)}(X_i | X_{pa(i), \mathcal{G}})
\prod_{i \in {B'}} f^{(\hat I)}(X_i | X_{pa(i), \mathcal{G}})
\prod_{i \in V \backslash V_{An}} f^{(\hat I)}(X_i | X_{pa(i), \mathcal{G}})\\
&=\prod_{i \in A'} f^{(\emptyset)}(X_i | X_{pa(i), \mathcal{G}})
\prod_{i \in C, pa_{\mathcal{G}}(i) \cap A' \neq \emptyset} f^{(\emptyset)}(X_i | X_{pa(i), \mathcal{G}})\\
&\prod_{i \in C, pa_{\mathcal{G}}(i) \cap A' = \emptyset} f^{(\hat I)}(X_i | X_{pa(i), \mathcal{G}})
\prod_{i \in {B'}} f^{(\hat I)}(X_i | X_{pa(i), \mathcal{G}})
\prod_{i \in V \backslash V_{An}} f^{(\hat I)}(X_i | X_{pa(i), \mathcal{G}})\\
\end{align*}

The second equality holds by the factorization of Equation (\ref{eq:int-factor}) because either $i \in A' $ or $i \in C ~|~pa_{\mathcal{G}(i)} \cap A' \neq \emptyset $ implies that $i$ is not targeted by the intervention on $\hat I$, i.e. $i \notin \hat I$. To see this, recall that $A'$ is separated from $\zeta_{\hat I}$ in $\mathcal{G}^{\mathcal{I}}$, which implies that $A'$ does not contain a child of $\zeta_{\hat I}$ in $\mathcal{G}^{\mathcal{I}}$ and is therefore not targeted by the intervention on $\hat I$. Likewise, $\{i \in C | pa_{\mathcal{G}(i)} \cap A' \neq \emptyset \}$ does not contain a child of $\zeta_{\hat I}$ in $\mathcal{G}^{\mathcal{I}}$ because otherwise $A'$ and $\zeta_{\hat I}$ would be d-connected in $\mathcal{G}^{\mathcal{I}}$ by conditioning on this node.

Using similar reasoning, it is easy to see that the parent sets of $A'$ and $\{i \in C | pa_{\mathcal{G}(i)} \cap A' \neq \emptyset \}$ with respect to $\mathcal{G}$ are subsets of $A' \cup C$; and the parent sets of $\{i \in C | pa_{\mathcal{G}(i)} \cap A' = \emptyset \}$ and ${B'}$ are subsets of $B' \cup C$. Therefore, we can write
\[
f^{(\hat I)}(X) = g_1(X_{A'}, X_C) g_2(X_{B'}, X_C; \hat I) \prod_{i \in V \backslash V_{An}} f^{(\hat I)}(X_i | X_{pa(i), \mathcal{G}})
\]
where 
\[g_1(X_{A'}, X_C) = \prod_{i \in A'} f^{(\emptyset)}(X_i | X_{pa(i), \mathcal{G}})
\prod_{i \in C, pa_{\mathcal{G}}(i) \cap A' \neq \emptyset} f^{(\emptyset)}(X_i | X_{pa(i), \mathcal{G}})\]
and 
\[g_2(X_{B'}, X_C; \hat I) = \prod_{i \in C, pa_{\mathcal{G}}(i) \cap A' = \emptyset} f^{(\hat I)}(X_i | X_{pa(i), \mathcal{G}})
\prod_{i \in {B'}} f^{(\hat I)}(X_i | X_{pa(i), \mathcal{G}})\]

Marginalizing out $X_{A'\backslash A}, X_{B'}$ and $X_{V_{An}}$ yields
\[
f^{(\hat I)}(X_A, X_C) = \hat g_1(X_A, X_C) \hat g_2(X_C; \hat I)
\]
where $\hat g_1(X_A, X_C) = \int_{X_{A'\backslash A}}g_1(X_{A'}, X_C)$ and $\hat g_2(X_C; \hat I) = \int_{X_{B'}} g_2(X_{B'}, X_C; \hat I)$. From here, it is easy to see that $f^{(\hat I)}(X_A | X_C) = \frac{g^{(\emptyset)}(X_A, X_C)}{\int_{X_A} g^{(\emptyset)}(X_A, X_C)}$ is invariant to $\hat I \in \{\emptyset, I\}$. 
\end{proof}

\begin{proof} [Proof of Lemma \ref{lem:i-markov-comp}]
Choose any disjoint $A, C \subset [p]$ and any $I \in \mathcal{I}$, and suppose $C \cup \zeta_{\mathcal{I} \backslash \{I\}}$ does not d-separate $A$ from $\{\zeta_{I}\}$ in $\mathcal{G}^{\mathcal{I}}$. To prove this lemma, it is sufficient to construct $f^{(\emptyset)}$ and $f^{(I)}$ such that they satisfy the $\hat{\mathcal{I}}$-Markov properties with respect to $\mathcal{G}^{\hat{\mathcal{I}}}$, where $\hat{\mathcal{I}} = \{\emptyset, I\}$, and $f^{(\emptyset)}(X_A|X_C) \neq f^{(I)}(X_A|X_C)$. 


To do this, we construct a subgraph $\mathcal{G}_{sub} = (V, E_{sub})$ that consists of a d-connected path, $P = \{p_1 \in I, p_2, \cdots, p_{k-1}, p_k \in A\}$ with $p_2, \cdots, p_{k-1} \notin I \cup A$, as well as the directed paths from colliders in $P$ to their nearest descendants in $C$. All other nodes that are not part of these paths are part of the subgraph but have no edges. We parameterize the set of conditional probability distributions, $\{f^{(\emptyset)}(X_i | X_{pa_{G_{sub}}(i)})\}_{i \in V}$, using linear structural equations with non-zero coefficients and independent Gaussian noise. Consider $f^{(\emptyset)}(X_{p_1} | X_{pa_{G_{sub}}(p_1)}) = \mathcal{N}(\sum_{j \in pa_{G_{sub}}(p_1)} c_{j,{p_1}} X_j, \sigma_{\emptyset}^2)$. To construct $f^{(I)}$, let $f^{(I)}(X_{p_1} | X_{pa_{G_{sub}}(p_1)}) = \mathcal{N}(\sum_{j \in pa_{G_{sub}}(p_1)} c_{j,{p_1}} X_j, \sigma_{I}^2)$ for some $\sigma_I \neq \sigma_{\emptyset}$, and let $f^{(I)}(X_i | X_{pa_{G_{sub}}(i)}) = f^{(\emptyset)}(X_i | X_{pa_{G_{sub}}(i)})$ for all $i \neq p_1$.

Note that these distributions factor over $\mathcal{G}$ according to Definition \ref{def:soft-int}, so by Proposition \ref{lem:i-markov}, they satisfy the $\hat{\mathcal{I}}$-Markov properties with respect to $\mathcal{G}^{\hat{\mathcal{I}}}$. { Furthermore, it is straightforward to show that we can write $X_{p_k} = N_{(\emptyset)} + S(X_C)$ under the model corresponding to $f^{(\emptyset)}$ and $X_{p_k} = N_{(I)} + S(X_C) $ under the model corresponding to $f^{(I)}$, where $N_{(\emptyset)} \sim \mathcal{N}(0, c\sigma_{\emptyset}^2)$ and $N_{(I)} \sim \mathcal{N}(0, c\sigma_{I}^2)$ for some constant $c$, and $S(X_C)$ is a Gaussian random variable independent of $N_{(\emptyset)}$ and $N_{(I)}$.} Since $\sigma_{\emptyset}^2 \neq \sigma_{I}^2$, it follows that we have $f^{(\emptyset)}(X_A|X_C) \neq f^{(I)}(X_A|X_C)$, as desired.
\end{proof}

\begin{proof} [Proof of Corollary \ref{cor:perfect-i-mec}]
If $\mathcal{G}_1$ and $\mathcal{G}_2$ are in the same perfect-$\mathcal{I}$-MEC, then $\mathcal{G}_1$ and $\mathcal{G}_2$ have the same skeleton and v-structures. Since $\mathcal{G}_1^{\mathcal{I}}$ and $\mathcal{G}_2^{\mathcal{I}}$ are constructed from $\mathcal{G}_1$ and $\mathcal{G}_2$ by adding the same set of vertices and edges, they must have the same skeleta, so we just need to show that they also have the same v-structures to prove they belong to the same $\mathcal{I}$-MEC. Suppose this is not the case. The only v-structures that can differ between $\mathcal{G}_1^{\mathcal{I}}$ and $\mathcal{G}_2^{\mathcal{I}}$ must involve $\mathcal{I}$-edges, since $\mathcal{G}_1$ and $\mathcal{G}_2$ have the same v-structures. Without loss of generality, suppose $\zeta_I \rightarrow i$ is part of a v-structure in $\mathcal{G}_1^{\mathcal{I}}$ but not in $\mathcal{G}_2^{\mathcal{I}}$. This could only occur if there were a neighbor $j \notin I$ with orientation $j \rightarrow i$ in $\mathcal{G}_1^{\mathcal{I}}$ and $i \rightarrow j$ in $\mathcal{G}_2^{\mathcal{I}}$. However, this contradicts the assumption that $\mathcal{G}_1$ and $\mathcal{G}_2$ belong to the same perfect-$\mathcal{I}$-MEC \cite{hauser12}, since removing the incoming edges of $i$ from $\mathcal{G}_1$ and $\mathcal{G}_2$ would result in graphs with different skeleta. Therefore, $\mathcal{G}_1^{\mathcal{I}}$ and $\mathcal{G}_2^{\mathcal{I}}$ must have the same v-structures.

Conversely, suppose that $\mathcal{G}_1$ and $\mathcal{G}_2$ are in the same $\mathcal{I}$-MEC. Then they must have the same skeleta and v-structures, and we just need to show that for any $I \in \mathcal{I}$, $\mathcal{G}_1$ and $\mathcal{G}_2$ have the same skeleton after removing the incoming edges of $i$ for all $i \in I$ \cite{hauser12}. Suppose this is not the case. This implies that for some $I \in \mathcal{I}$ and some $i \in I$, there is an edge between $i$ and another vertex $j \notin I$ that is removed in $\mathcal{G}_1$ but not in $\mathcal{G}_2$. The orientation of this edge must be $j \rightarrow i$ in $\mathcal{G}_1$ and $i \rightarrow j$ in $\mathcal{G}_2$. But this would mean that $j \rightarrow i$ and $\zeta_I \rightarrow i$ form a v-structure in $\mathcal{G}_1^{\mathcal{I}}$ but not in $\mathcal{G}_2^{\mathcal{I}}$, which is a contradiction to Theorem \ref{the:i-mec}. Therefore, $\mathcal{G}_1$ and $\mathcal{G}_2$ must belong to the same perfect-$\mathcal{I}$-MEC. 
\end{proof}

\subsection{Proofs from Section \ref{sec_no_obs_data}}

The following definition formalizes the notion of relabeling the datasets and intervention targets:

\begin{definition} \label{def:reindex}
Let $\{f^{(I)}\}_{I \in \mathcal{I}}$ be a set of interventional distributions. Let $J \in \mathcal{I}$ be a particular intervention target. The corresponding \emph{$J$-observation target set} is defined as $\tilde{\mathcal{I}}_J:=\{\emptyset, \{I \cup J\}_{I \in \mathcal{I}, I \neq J} \}$. The relabeled set of interventional distributions is denoted $\{\tilde{f}_J^{(I)}\}_{I \in \tilde{\mathcal{I}}_J}$, with $\tilde{f}_J^{(\emptyset)} := f^{(J)}$ and $\tilde{f}_J^{(I \cup J)} := f^{(I)}, ~\forall I \in \mathcal{I}$, $I \neq J$.
\end{definition}

Notice that $\{\tilde{f}_J^{(I)}\}_{I \in \tilde{\mathcal{I}}_J}$ contains the same distributions as $\{f^{(I)}\}_{I \in \mathcal{I}}$ but is reindexed to treat $f^{(J)}$ as the observational distribution and $\{f^{(I)}\}_{I \neq J}$ as distributions obtained under interventions on $I \cup J$. This relabeling is justified by the following lemma:

\begin{lemma} \label{lem:reindex} $\{f^{(I)}\}_{I \in \mathcal{I}} \in \mathcal{M}_{\mathcal{I}}(\mathcal{G})$ if and only if $\{\tilde{f}_J^{(I)}\}_{I \in \tilde{\mathcal{I}}_J} \in \mathcal{M}_{\tilde{\mathcal{I}}_J}(\mathcal{G})$ for all $J \in \mathcal{I}$.
\end{lemma}

\begin{proof} [Proof of Lemma \ref{lem:reindex}]
To prove the ``only if" direction, suppose $\{f^{(I)}\}_{I \in \mathcal{I}} \in \mathcal{M}_{\mathcal{I}}(\mathcal{G})$. It follows straight from Definition \ref{def:reindex} that $\tilde{f}_J^{(\emptyset)} = f^{(J)}$ for every $J \in \mathcal{I}$. Since $f^{(J)}$ is Markov with respect to $\mathcal{G}$, so is $\tilde{f}_J^{(\emptyset)}$, and hence $\tilde{f}_J^{(\emptyset)}$ can be factored according to Equation (\ref{eq:int-factor}) with the observational distribution set to $\tilde{f}_J^{(\emptyset)}$. So it remains to show that for every $J \in \mathcal{I}$ and any $I \neq J$, $\tilde{f}_J^{(I \cup J)}$ factorizes according to Equation (\ref{eq:int-factor}) with the observational distribution set to $\tilde{f}_J^{(\emptyset)}$. Then we have
\[
\begin{split}
\tilde{f}_J^{(I \cup J)}(X) &= f^{(I)}(X)  \quad \quad \quad\quad \quad \quad \quad \quad \quad\quad \quad \quad \quad \quad \quad \text{(by Definition \ref{def:reindex})}\\
&= \prod_{i \notin I} f^{(\emptyset)}(X_i | X_{pa(i)}) \prod_{i \in I} f^{(I)}(X_i | X_{pa(i)}) \quad \quad \quad \text{(by Lemma \ref{lem:M})} \\
&= \prod_{i \notin I,J} f^{(\emptyset)}(X_i | X_{pa(i)}) \prod_{i \notin I, i \in J} f^{(\emptyset)}(X_i | X_{pa(i)}) \prod_{i \in I} f^{(I)}(X_i | X_{pa(i)})\\
&= \prod_{i \notin I \cup J} \tilde{f}_J^{(\emptyset)}(X_i | X_{pa(i)}) \prod_{i \in I \cup J} \tilde{f}_J^{(I \cup J)}(X_i | X_{pa(i)})
\end{split}
\]
where the last equality holds because $\tilde{f}_J^{(\emptyset)}(X_i | X_{pa(i)}) = f^{(J)}(X_i | X_{pa(i)}) = f^{(\emptyset)}(X_i | X_{pa(i)})$ when $i \notin J$, and by relabeling the conditional distributions in the last two product terms as $\tilde{f}_J^{(I \cup J)}$. By Lemma \ref{lem:M}, it follows that $\{\tilde{f}_J^{(I)}\}_{I \in \tilde{\mathcal{I}}_J} \in \mathcal{M}_{\tilde{\mathcal{I}}_J}(\mathcal{G})$.

To prove the converse, we show how to construct the observational distribution $f^{(\emptyset)}$ such that $f^{(I)}$ can be factored over $\mathcal{G}$ according to Equation (\ref{eq:int-factor}) for all $I \in \mathcal{I}$. For every $i \in \mathcal{V}$, let $f^{(\emptyset)}(X_i | X_{pa_{\mathcal{G}}(i)}) = \tilde{f}_I^{(\emptyset)}(X_i | X_{pa_{\mathcal{G}}(i)})$ for some $I \in \mathcal{I}$ such that $i \notin I$. The existence of such an $I$ is guaranteed by the assumption that $\mathcal{I}$ is a conservative set of targets. Furthermore, $\tilde{f}_I^{(\emptyset)}(X_i | X_{pa_{\mathcal{G}}(i)})$ is unique; if there are multiple targets that satisfy this requirement (i.e. $\exists J \in \mathcal{I} ~s.t.~ i \notin J$ and  $J\neq I$), we always have $\tilde{f}_I^{(\emptyset)}(X_i | X_{pa_{\mathcal{G}}(i)}) = \tilde{f}_J^{(\emptyset)}(X_i | X_{pa_{\mathcal{G}}(i)})$, since
\[
\tilde{f}_J^{(\emptyset)}(X_i | X_{pa_{\mathcal{G}}(i)})
= \tilde{f}_I^{(I \cup J)}(X_i | X_{pa_{\mathcal{G}}(i)})
= \tilde{f}_I^{(\emptyset)}(X_i | X_{pa_{\mathcal{G}}(i)})
\]
for $i \notin I \cup J$. The first equality follows by Definition \ref{def:reindex}, and the second equality follows since by hypothesis, $\{\tilde{f}_I^{(K)}\}_{K \in \tilde{\mathcal{I}}_{I}} \in \mathcal{M}_{\tilde{\mathcal{I}}_{I}}(\mathcal{G})$. 
Thus, we have defined $f^{(\emptyset)}$ such that $f^{(I)}$ can be factored over $\mathcal{G}$ according to Equation (\ref{eq:int-factor}) for all $I \in \mathcal{I}$. This proves by Lemma \ref{lem:M} that $\{f^{(I)}\}_{I \in \mathcal{I}} \in \mathcal{M}_{\mathcal{I}}(\mathcal{G})$.
\end{proof}

\begin{proof}[Proof of Theorem \ref{th:imec-2}]
We first prove the ``only if" direction. 
Suppose $\mathcal{G}_1^{\tilde{\mathcal{I}}_J}$ and $\mathcal{G}_2^{\tilde{\mathcal{I}}_J}$ do not have the same skeleton and v-structures for some $J \in \mathcal{I}$. If $\mathcal{G}_1, \mathcal{G}_2$ do not have the same v-structures and skeletons, then they do not belong to the same MEC and it is straightforward to see that $\mathcal{M}_{\mathcal{I}}(\mathcal{G}_1)\neq\mathcal{M}_{\mathcal{I}}(\mathcal{G}_2)$. Otherwise there exists $I \in \tilde{\mathcal{I}}_J$ and $j \in [p]$ such that $\zeta_I \rightarrow j$ is part of a v-structure in one $\tilde{\mathcal{I}}_J$-DAG and not the other. Suppose without loss of generality that $\zeta_I \rightarrow j$ is part of a v-structure in $\mathcal{G}_1^{\tilde{\mathcal{I}}_J}$ but not in $\mathcal{G}_2^{\tilde{\mathcal{I}}_J}$. Then $j$ has a neighbor $k\in[p]\setminus\{j\}$ with orientation $k \rightarrow j$ in $\mathcal{G}_1^{\tilde{\mathcal{I}}_J}$ and $j \rightarrow k$ in $\mathcal{G}_2^{\tilde{\mathcal{I}}_J}$; $k$ and $\zeta_I$ given $pa_{\mathcal{G}_2^{\tilde{\mathcal{I}}_J}}(k)$ are d-connected in $\mathcal{G}_1^{\tilde{\mathcal{I}}_J}$  but d-separated in $\mathcal{G}_2^{\tilde{\mathcal{I}}_J}$. Similar to the proof of Lemma \ref{lem:i-markov-comp}, one can construct $\{f^{(I)}\}_{I \in \mathcal{I}}$ such that $\{\tilde{f}_J^{(I)}\}_{I \in \tilde{\mathcal{I}}_J}$ satisfies the $\tilde{\mathcal{I}}_J$-Markov property with respect to $\mathcal{G}_1^{\tilde{\mathcal{I}}_J}$ for all $J \in {\mathcal{I}}$ but not with respect to $\mathcal{G}_2^{\tilde{\mathcal{I}}_J}$ for some $J \in {\mathcal{I}}$. It follows from Lemma \ref{lem:reindex} that $\{f^{(I)}\}_{I \in \mathcal{I}} \in \mathcal{M}_{\mathcal{I}}(\mathcal{G}_1)$ but $\{f^{(I)}\}_{I \in \mathcal{I}} \notin \mathcal{M}_{\mathcal{I}}(\mathcal{G}_2)$, so $\mathcal{M}_{\mathcal{I}}(\mathcal{G}_1)\neq\mathcal{M}_{\mathcal{I}}(\mathcal{G}_2)$.
The ``if" direction follows by applying Theorem \ref{the:i-mec} to $\mathcal{G}_1^{\tilde{\mathcal{I}}_J}$ and $\mathcal{G}_2^{\tilde{\mathcal{I}}_J}$ for every $J \in \mathcal{I}$, followed by Lemma \ref{lem:reindex}.
\end{proof}

\section{Proofs from Section \ref{sec_alg}}
\subsection{Proof of Theorem \ref{the:consistency}}

In this section, we work up to the proof of Theorem \ref{the:consistency}. To do this, we first cover some basic results on the consistency of GSP. Let $\mathcal{G}$ be a DAG and let $\mathcal{H}$ be an independence map (I-map) of $\mathcal{G}$, meaning that all independences implied by $\mathcal{H}$ are satisfied by $\mathcal{G}$ (i.e. $\mathcal{G} \leq \mathcal{H}$). Chickering (2002) showed that there exists a sequence of covered edge reversals and edge additions resulting in a sequence of DAGs, $\mathcal{G}_0, \mathcal{G}_1, \cdots, \mathcal{G}_{\tau}$ such that
\[
\mathcal{G} = \mathcal{G}_0 \leq \mathcal{G}_1 \leq \cdots \leq \mathcal{G}_{\tau} = \mathcal{H}
\]
Furthermore, \citet{solus17} showed that for any $\mathcal{G}$ and $\mathcal{H}$, there exists such a Chickering sequence in which one sink node of $\mathcal{H}$ is fixed at a time. The following lemma connects this sequence over DAGs to a sequence over the topological orderings of the nodes.

\begin{lemma}
Let $\mathcal{G}_{i_1}, \cdots, \mathcal{G}_{i_p}$ be a subsequence of the Chickering sequence where one sink is fixed at a time, and let $\mathcal{G}_{i_j}$ be the first DAG in which the $j$th sink node is fixed, i.e. the sequence of DAGs from $\mathcal{G}_{i_{j-1}}$ to $\mathcal{G}_{i_j}$ involve covered edge reversals and edge additions required to resolve sink node $j$. Furthermore, let $\Pi(\mathcal{G})$ denote the set of topological orderings that are consistent with $\mathcal{G}$. Then for any $\pi_{i_{j-1}} \in \Pi(\mathcal{G}_{i_{j-1}})$ in which the last $j-1$ nodes correspond to the first $j-1$ fixed sink nodes, there exists a sequence of orderings $\pi_{i_{j-1}}, \cdots, \pi_{i_j}$ with $\pi_{k} \in \Pi(\mathcal{G}_k)$ such that the $j$th sink node moves only to the right, stopping in the $j$th position from the end, and the relative ordering of the other nodes remain unchanged.
\end{lemma}
\begin{proof}
The correctness of this lemma follows directly from Lemma 13 of \citet{solus17}.
\end{proof}
The following corollary is an immediate consequence of this lemma.
\begin{corollary} \label{lem:chickering_order}
For any DAG $\mathcal{G}$ over vertex set $[p]$ and any I-map $\mathcal{H}$, there exists a sequence of topological orderings
\[
\pi_0 \in \Pi(\mathcal{G}_0), \pi_1 \in \Pi(\mathcal{G}_1), \cdots, \pi_{\tau} \in \Pi(\mathcal{G}_{\tau})
\]
with $\mathcal{G}_0 = \mathcal{G}$ and $\mathcal{G}_{\tau} = \mathcal{H}$ corresponding to a Chickering sequence in which we fix the order of the nodes in reverse starting from the last node in $\pi_{\tau}$. Specifically, the last node in $\pi_{\tau}$ is moved to the right until it is in the $p$-th position, then the second-last node in $\pi_{\tau}$ to the right until it is in the $(p-1)$-th position, etc. until all nodes are in the order given by $\pi_{\tau}$.
\end{corollary}

Using this result, we now state the following lemma, which is useful in the proof of consistency of the algorithm.

\begin{lemma} \label{lem:directed_path}
For any permutation $\pi$, there exists a list of covered arrow reversals from $\mathcal{G}_{\pi}$ to the true DAG $\mathcal{G}_{\pi^*}$ such that (1) the number of edges is weakly decreasing:
\[
\mathcal{G}_{\pi} = \mathcal{G}_{\pi^0} \geq \mathcal{G}_{\pi^1} \geq \cdots \geq
\mathcal{G}_{\pi^{m-1}} \geq \mathcal{G}_{\pi^m}\geq \cdots \geq
\mathcal{G}_{\pi^{M-1}} \geq \mathcal{G}_{\pi^M} = \mathcal{G}_{\pi^*}
\]
and (2) if $i \rightarrow j$ is reversed from $\mathcal{G}_{\pi^{m-1}}$ to $\mathcal{G}_{\pi^m}$, then there is no directed path from $i$ to $j$ in $\mathcal{G}_{\pi^*}$.

\end{lemma}

\begin{proof}
It is sufficient to show that there exists a Chickering sequence from $\mathcal{G}_{\pi^*}$ to $\mathcal{G}_{\pi}$ such that if $i$ is an ancestor of $j$ in $\mathcal{G}_{\pi^*}$, then $j \rightarrow i$ is not reversed in the sequence. For all ancestor-descendant pairs $i,j$ in $\mathcal{G}_{\pi^*}$, $i$ precedes $j$ in all orderings belonging to $\Pi(G_{\pi^*})$. By Corollary \ref{lem:chickering_order}, there exists a Chickering sequence from $\mathcal{G}_{\pi^*}$ to $\mathcal{G}_{\pi}$ and a corresponding sequence of orderings such that no node ever moves from the left of its ancestor to the right of its ancestor. Specifically, for any ancestor-descendant pair $i,j$, either $j$ is fixed before $i$ and their relative ordering never changes, or $i$ is fixed first before $j$ and moves from the left of $j$ to the right of $j$ once. It follows that there is no edge $j \rightarrow i$ reversed in the Chickering sequence from $\mathcal{G}_{\pi^*}$ to $\mathcal{G}_{\pi}$. 
\end{proof}

In turn, Lemma \ref{lem:directed_path} allows us to prove the existence of a greedy path from $\mathcal{G}_{\pi}$ to the true DAG $\mathcal{G}_{\pi^*}$ by reversing $\mathcal{I}$-covered edges.

\begin{lemma} \label{lem:i-cov}
For any permutation $\pi$, there exists a list of $\mathcal{I}$-covered arrow reversals from $\mathcal{G}_{\pi}$ to the true DAG $\mathcal{G}_{\pi^*}$ such that the number of edges is weakly decreasing.
\end{lemma}

\begin{proof}
From Lemma \ref{lem:directed_path}, we know that there exists a sequence of covered arrow reversals from $\mathcal{G}_{\pi}$ to $\mathcal{G}_{\pi^*}$ in which the number of edges is weakly decreasing; and that this sequence has the property that if arrow $i \rightarrow j$ is reversed from $\mathcal{G}_{\pi^{m-1}}$ to $\mathcal{G}_{\pi^m}$, then there is no directed path from $i$ to $j$ in $\mathcal{G}_{\pi^*}$. It remains to be shown that $i \rightarrow j$ is $\mathcal{I}$-covered. Suppose $\{i\} \in \mathcal{I}$ and let $\mathcal{G}^{(\mathcal{I})}_{\pi^*}$ denote the $\mathcal{I}$-DAG of $\mathcal{G}_{\pi^*}$ (Definition \ref{def:i-dag}). Note that $\{\zeta_{\{i\}}\}$ is d-separated from $j$ in $\mathcal{G}^{(\mathcal{I})}_{\pi^*}$ since there is no directed path from $i$ to $j$. Therefore, $f(X_j)$ is invariant to $I \in \{\emptyset, \{i\}\}$ by the $\mathcal{I}$-Markov property (Definition \ref{def:i-markov}) and Proposition \ref{lem:i-markov}. It follows that $i \rightarrow j$ is $\mathcal{I}$-covered in $\mathcal{G}_{\pi^{m-1}}$. If $\{i\} \notin \mathcal{I}$, then the result is trivial as $i \rightarrow j$ is $\mathcal{I}$-covered as long as it is covered.
\end{proof}

The following lemma proves the correctness of using $\mathcal{I}$-contradictory arrows as the secondary search criterion; essentially, it states that when $\mathcal{G}_{\pi}$ is in the same MEC but not the same $\mathcal{I}$-MEC as $\mathcal{G}_{\pi^*}$, then $\mathcal{G}_{\pi}$ has more $\mathcal{I}$-contradictory arrows than $\mathcal{G}_{\pi^*}$.

\begin{lemma} \label{lem:i-cont}
For any permutation $\pi$ such that $\mathcal{G}_{\pi}$ and $\mathcal{G}_{\pi^*}$ are in the same MEC, there exists a list of $\mathcal{I}$-covered arrow reversals from $\mathcal{G}_{\pi}$ to the true DAG $\mathcal{G}_{\pi^*}$
\[
\mathcal{G}_{\pi} = \mathcal{G}_{\pi^0} \geq \mathcal{G}_{\pi^1} \geq \cdots \geq
\mathcal{G}_{\pi^{m-1}} \geq \mathcal{G}_{\pi^m}\geq \cdots \geq
\mathcal{G}_{\pi^{M-1}} \geq \mathcal{G}_{\pi^M} = \mathcal{G}_{\pi^*}
\]
such that the number of arrows is non-increasing and for all $m$, if $\mathcal{G}_{\pi^{m-1}}$ and $\mathcal{G}_{\pi^m}$ are not in the same $\mathcal{I}$-MEC, then $\mathcal{G}_{\pi^m}$ is produced from $\mathcal{G}_{\pi^{m-1}}$ by the reversal of an $\mathcal{I}$-contradictory arrow.
\end{lemma}

\begin{proof}
From Lemma \ref{lem:i-cov}, we know there exists a sequence of $\mathcal{I}$-covered arrow reversals from $\mathcal{G}_{\pi}$ to $\mathcal{G}_{\pi^*}$ in which the number of edges is weakly decreasing, with the property that if $i \rightarrow j$ is reversed from $\mathcal{G}_{\pi^{m-1}}$ to $\mathcal{G}_{\pi^m}$, then $i$ is not an ancestor of $j$ in $\mathcal{G}_{\pi^*}$. 

Suppose the arrow $i \rightarrow j$ is reversed from $\mathcal{G}_{\pi^{m-1}}$ to $\mathcal{G}_{\pi^m}$. Since $\mathcal{G}_{\pi^{m-1}}$ and $\mathcal{G}_{\pi^*}$ are in the same MEC and $i$ is not an ancestor of $j$, this implies that $j \rightarrow i$ is in $\mathcal{G}_{\pi^*}$. Since $\mathcal{G}_{\pi^{m-1}}$, $\mathcal{G}_{\pi^m}$ are not in the same $\mathcal{I}$-MEC, then we must have $\mathcal{I}_{i\backslash j} \cup \mathcal{I}_{i\backslash j} \neq \emptyset $. Now, let $\mathcal{G}^{(\mathcal{I})}_{\pi^*}$ denote the $\mathcal{I}$-DAG of $\mathcal{G}_{\pi^*}$ (Definition \ref{def:i-dag}), and consider the following cases:

(1) $\mathcal{I}_{i \backslash j} \neq \emptyset$. Then there exists a subset $S \subset ne_{\mathcal{G}}(j) \backslash \{i\}$ that d-separates $\zeta_{\mathcal{I}_{i \backslash j}}$ from $j$ in $\mathcal{G}^{(\mathcal{I})}_{\pi^*}$. By the $\mathcal{I}$-Markov property (Definition \ref{def:i-markov}) and Proposition \ref{lem:i-markov}, $f^{(\emptyset)}(X_j|X_S) = f^{(I)}(X_j|X_S)$ for all $I \in \mathcal{I}_{i \backslash j}$. 

(2) $\mathcal{I}_{j \backslash i} \neq \emptyset$. Then for any subset $S \subset ne_{\mathcal{G}}(i) \backslash \{j\}$, $\zeta_{\mathcal{I}_{j \backslash i}}$ is d-connected to $i$ in $\mathcal{G}^{(\mathcal{I})}_{\pi^*}$. By Assumption \ref{ass:2}, $f^{(\emptyset)}(X_i|X_S) \neq f^{(I)}(X_i|X_S)$ for some $I \in \mathcal{I}_{j \backslash i}$.

(3) $\{i\} \in \mathcal{I}$. Then $\zeta_{\{i\}}$ is d-separated from $j$ in $\mathcal{G}^{(\mathcal{I})}_{\pi^*}$. Therefore, $f^{\{i\}}(X_j)=f^{\emptyset}(X_j)$ by the $\mathcal{I}$-Markov property (Definition \ref{def:i-markov}) and Proposition \ref{lem:i-markov}.

(4) $\{j\} \in \mathcal{I}$. Then $\zeta_{\{j\}}$ is not d-separated from $i$ in $\mathcal{G}^{(\mathcal{I})}_{\pi^*}$. Therefore, $f^{\{j\}}(X_i)\neq f^{\emptyset}(X_i)$ by Assumption \ref{ass:1}.

These are the defining properties of $\mathcal{I}$-contradictory edges. Therefore, the arrow $i \rightarrow j$ is $\mathcal{I}$-contradictory in $\mathcal{G}_{\pi^{m-1}}$.
\end{proof}

\begin{proof}[Proof of Theorem \ref{the:consistency}] 
This follows directly from Lemmas \ref{lem:i-cov} and \ref{lem:i-cont}.
\end{proof}

\subsection{Pooling Data for CI Testing}

The following proposition gives sufficient conditions under which CI relations hold when the data come from a mixture of interventional distributions:

\begin{proposition}\label{pro:pool}
Let $\{f^{(I)}\}_{I \in \mathcal{I}} \in \mathcal{M}_{\mathcal{I}}(\mathcal{G})$ for a DAG $\mathcal{G} = ([p], E)$ and intervention targets $\mathcal{I}$ s.t. $\emptyset \in \mathcal{I}$. For some $\mathcal{I}_s \subset \mathcal{I}$ and some disjoint $A, B, C \subset [p]$, suppose that $C \cup \zeta_{\mathcal{I}\backslash \mathcal{I}_s}$ d-separates $A$ from $B \cup \zeta_{\mathcal{I}_s}$ in $\mathcal{G}^{\mathcal{I}}$. Then $X_A \indep X_B \mid X_C$ under the distribution $X \sim \sum_{I \in \{\emptyset\} \cup \mathcal{I}_s} \alpha_{I}f^{(I)}$, for any $\alpha_{I} \in (0,1)$ s.t. $\sum_{I \in \{\emptyset\} \cup \mathcal{I}_s} \alpha_{I} = 1 $.
\end{proposition}

Proposition \ref{pro:pool} can be used to derive a set of checkable conditions on $\mathcal{G}_{\pi}$ to determine whether each interventional dataset $I \in \mathcal{I}$ can be pooled with observational data to test $X_i \indep X_k \mid X_{\an_{\mathcal{G}_\pi}(i) \setminus \{k\}}$ for $k \in \pa_{\mathcal{G}_\pi}(i)$. 

\begin{corollary} \label{cor:pool}
Suppose we want to test $X_i \indep X_k \mid X_{\an_{\mathcal{G}_\pi}(i) \setminus \{k\}}$ for some $k \in \pa_{\mathcal{G}_\pi}(i)$. Let $\mathcal{I}_s \subset \mathcal{I}$ be interventional targets such that the following two conditions hold for every $j \in I \in \mathcal{I}_s$:

(1) $j = i$ or $j$ is neither a descendant nor an ancestor of $i$;

(2) $\pi(k) > \pi(j)$ and $k$ is not a parent of $j$; or $\pi(j) > \pi(k)$ and $j$ is not an ancestor of $k$,

where all relations are being considered with respect to $\mathcal{G}_{\pi}$, and $\pi(i)$ denotes the index of $i$ in $\pi$. Then under the faithfulness assumption, $X_i \indep X_k \mid X_{\an_{\mathcal{G}_\pi}(i) \setminus \{k\}}$ under $X \sim f^{\emptyset}$ if and only if this CI relation also holds under $X \sim \sum_{I \in \{\emptyset\} \cup \mathcal{I}_s} \alpha_{I}f^{(I)}$, where $\alpha_{I} \in (0,1)$ and $\sum_{I \in \{\emptyset\} \cup \mathcal{I}_s} \alpha_{I} = 1$.

\end{corollary}

\begin{proof}
If $X_i \nindep X_k \mid X_{\an_{\mathcal{G}_\pi}(i) \setminus \{k\}}$ under $X \sim f^{\emptyset}$, then this CI relation will clearly not hold under $X \sim \sum_{I \in \{\emptyset\} \cup \mathcal{I}_s} \alpha_{I}f^{(I)}$, thereby implying the ``if" direction. It remains to prove the ``only if" direction, i.e. that $X_i \indep X_k \mid X_{\an_{\mathcal{G}_\pi}(i) \setminus \{k\}}$ under $X \sim f^{\emptyset}$ implies conditional independence under $X \sim \sum_{I \in \{\emptyset\} \cup \mathcal{I}_s} \alpha_{I}f^{(I)}$.

We first consider the case where $j\neq i$ and $j$ is neither a descendant nor an ancestor of $i$. By the faithfulness assumption, $X_i \indep X_k \mid X_{\an_{\mathcal{G}_\pi}(i) \setminus \{k\}}$ implies that $i$ and $k$ are d-separated by $\an_{\mathcal{G}_\pi}(i) \setminus \{k\}$ in the true DAG $\mathcal{G}_*$. Since $\mathcal{G}_{\pi}$ is an independence map of $\mathcal{G}_*$, it follows from condition (2) that for any $j \in I \in \mathcal{I}_s$, $j$ and $k$ are d-separated by $\an_{\mathcal{G}_\pi}(j)\setminus \{k\}$ in $\mathcal{G}_*$. In addition, since $j$ is neither a descendant nor an ancestor of $i$, then $j$ and $k$ are also d-separated by $\an_{\mathcal{G}_\pi}(i) \setminus \{k\}$ in $\mathcal{G}_*$. 

If $i=j\in I \in \mathcal{I}_s$, then $k$ and $\{ i \} \cup \zeta_{\mathcal{I}_s}$ are d-separated in $\mathcal{G}_*^{\mathcal{I}}$ by $\zeta_{\mathcal{I} \setminus \mathcal{I}_s} \cup \an_{\mathcal{G}_\pi}(i) \setminus \{k\}$. It then follows from Proposition \ref{pro:pool} that $X_i \indep X_k \mid X_{\an_{\mathcal{G}_\pi}(i) \setminus \{k\}}$ when $X \sim \sum_{I \in \{\emptyset\} \cup \mathcal{I}_s} \alpha_{I}f^{(I)}$. 
\end{proof}

\begin{proof}[Proof of Proposition \ref{pro:pool}]
Similar to the proof of the second part of Proposition \ref{lem:i-markov}, it can be shown that for any disjoint $A, B, C \subset [p]$ and any $I \in \mathcal{I}$ such that $C \cup \zeta_{\mathcal{I} \backslash \{I\}}$ d-separates $A$ from $\{\zeta_{I}\}$ in $\mathcal{G}^{\mathcal{I}}$, we have
\[
f^{(I)}(X) = g_1(X_{A'}, X_C) g_2(X_{B'}, X_C; I) \prod_{i \in V \backslash V_{An}} f^{(I)}(X_i | X_{pa(i), \mathcal{G}})
\]
where 
\[g_1(X_{A'}, X_C) = \prod_{i \in A'} f^{(\emptyset)}(X_i | X_{pa(i), \mathcal{G}})
\prod_{i \in C, pa_{\mathcal{G}}(i) \cap A' \neq \emptyset} f^{(\emptyset)}(X_i | X_{pa(i), \mathcal{G}})\]
and 
\[g_2(X_{B'}, X_C; I) = \prod_{i \in C, pa_{\mathcal{G}}(i) \cap A' = \emptyset} f^{(I)}(X_i | X_{pa(i), \mathcal{G}})
\prod_{i \in {B'}} f^{(I)}(X_i | X_{pa(i), \mathcal{G}})\]

where $V_{An}$ is the ancestral set of $A \cup B \cup C$, $A'$ is the largest subset of $V_{An}$ that is d-separated from $B$ and $\zeta_{I}$ given $C$, and $B' = V_{An} \setminus (A' \cup C)$. Noting that $B \subset B'$, we marginalize out $X_{A' \setminus A}$,$X_{B' \setminus B}$ and $X_{B' \setminus B}$, which yields
\[
f^{(I)}(X_A, X_C) = \hat g_1(X_A, X_C) \hat g_2(X_B, X_C; I)
\]
The mixture of distributions over all $I \in \mathcal{I}_s$ is therefore,
\[
\sum_{I \in \mathcal{I}_s} \alpha_{I} f^{(I)}(X_A, X_C) = \hat g_1(X_A, X_C) \sum_{I \in \mathcal{I}_s} \alpha_{I} \hat g_2(X_B, X_C; I)
\]
which factors into separate functions over $X_A$ and $X_B$. Therefore, $X_A \indep X_B | X_C$ when $X$ is sampled from this mixture of distributions.
\end{proof}

\section{Additional simulation results}

\subsection{IGSP vs. perfect-IGSP}

As described in the main text, for each simulation, we sampled $100$ DAGs from an Erd\"os-Renyi random graph model with an average neighborhood size of $1.5$ and $p \in \{10, 20\}$ nodes. The data for each DAG $\mathcal{G}^*$ was generated using a linear structural equation model with independent Gaussian noise: $X = AX + \epsilon$, where $A$ is an upper-triangular matrix with edge weights $A_{ij} \neq 0$ if and only if $i \rightarrow j$, and $\epsilon \sim \mathcal{N}(0,Id)$. For $A_{ij} \neq 0$, the edge weights were sampled uniformly from $[-1, -0.25] \cup [0.25, 1]$ to ensure that they are bounded away from zero. We simulated perfect interventions on $i$ by setting the column $A_{,i} = 0$; inhibiting interventions by decreasing $A_{,i}$ by a factor of $10$; and imperfect interventions with a success rate of $\alpha=0.5$. Here, the results are shown for 10-node graphs in which interventions were performed on all single-variable targets (Figure \ref{fig:sim_igsp_extra}), or all pairs of multiple-variable targets (Figure \ref{fig:sim_igsp_extra}).

IGSP performed better on single-variable interventions than on multi-variable interventions (Figure \ref{fig:sim_igsp_extra}). This is expected based on the discussion on Definition \ref{def:i-cont-edge}; IGSP requires fewer invariance tests when the data come from single-variable interventions. In contrast, perfect-IGSP \cite{wang17} performs similarly between single-variable and multi-variable interventions; by assuming perfect interventions, perfect-IGSP avoids multiple hypothesis testing when there are multi-variable interventions. 

\begin{figure*}[t] 
\center
\subfigure[]{
\includegraphics[scale=0.3]{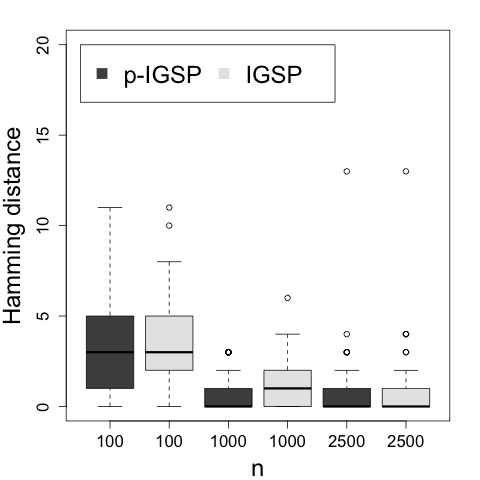}
}
\subfigure[]{
\includegraphics[scale=0.3]{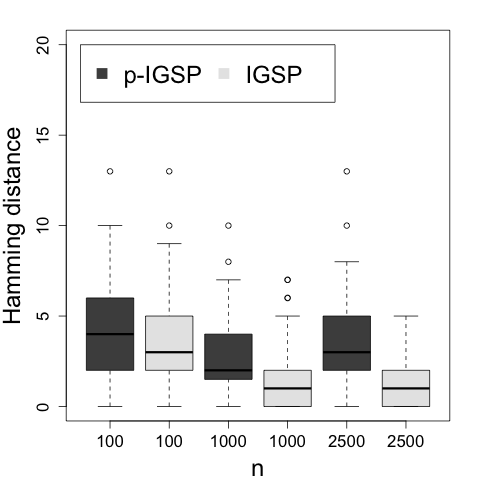}
}
\subfigure[]{
\includegraphics[scale=0.3]{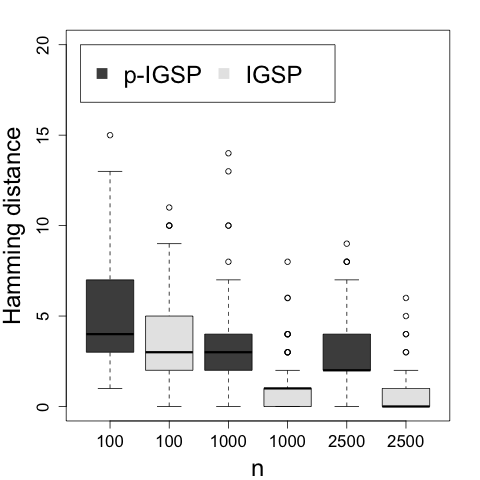}
}
\subfigure[]{
\includegraphics[scale=0.3]{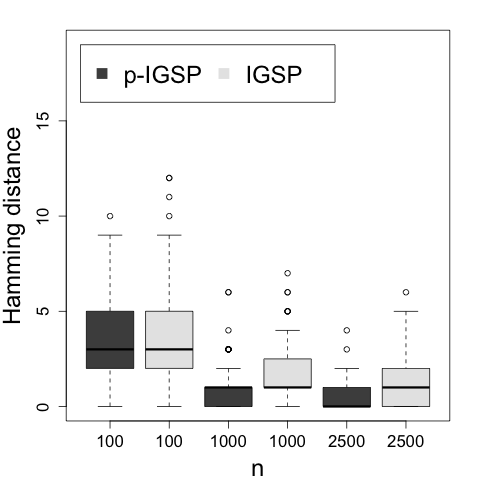}
}
\subfigure[]{
\includegraphics[scale=0.3]{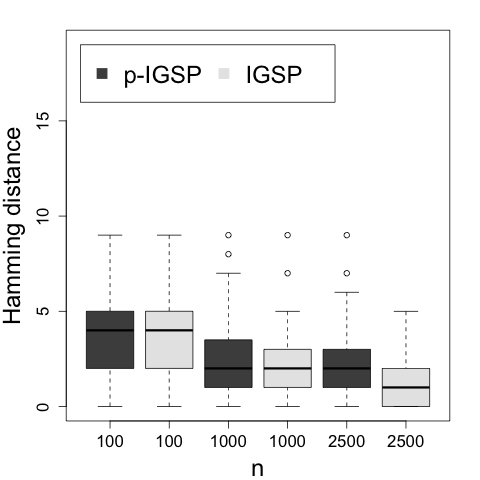}
}
\subfigure[]{
\includegraphics[scale=0.3]{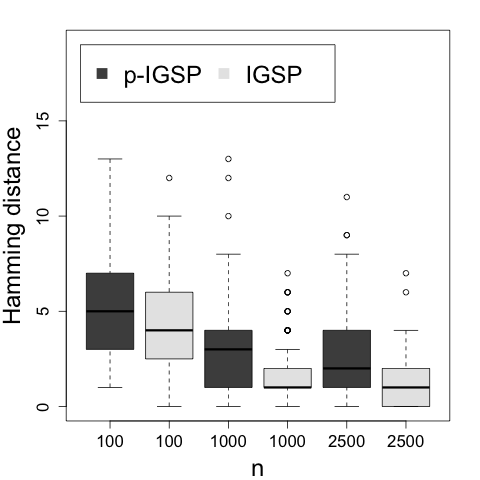}
}
\caption{Distributions of Hamming distances of recovered DAGs using IGSP and perfect-IGSP (p-IGSP) for 10-node graphs under single-variable (a) perfect, (b) imperfect, and (c) inhibitory interventions and multi-variable (d) perfect, (e) imperfect, and (f) inhibitory interventions}
\label{fig:sim_igsp_extra}
\end{figure*}

\subsection{Pooling}
Corollary \ref{cor:pool} described testable conditions under which CI tests can be performed over pooled observational and interventional data in a provably correct way. Here we show that the simple heuristic of pooling all of the datasets for all the CI tests is also effective for improving the performance of IGSP, particularly when the sample sizes are limited. The simulations of Figure \ref{fig:sim_pool_extra} compare IGSP to a heuristic version of IGSP, in which all of the data is pooled. However, the limitation of this method is that it is obviously not consistent in the limit of $n \rightarrow \infty$.

\begin{figure*}[t] 
\center
\subfigure[]{
\includegraphics[scale=0.3]{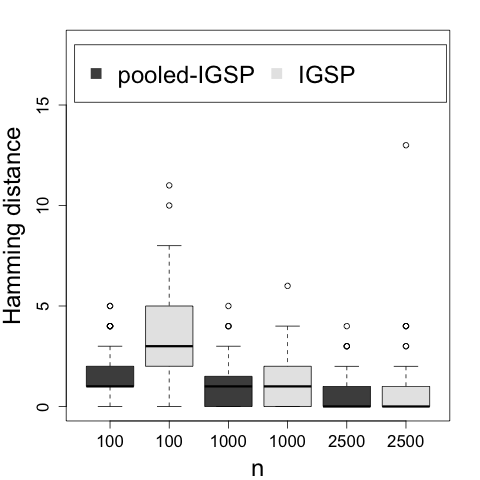}
}
\subfigure[]{
\includegraphics[scale=0.3]{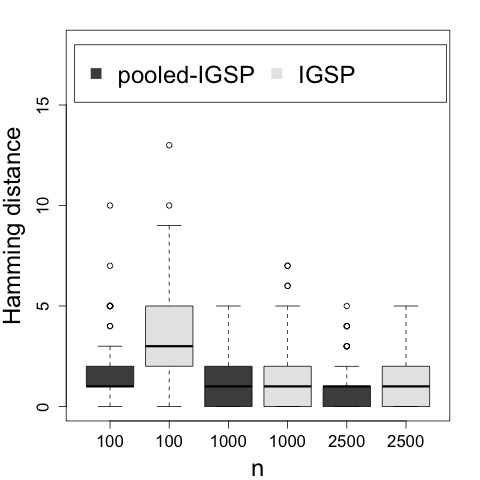}
}
\subfigure[]{
\includegraphics[scale=0.3]{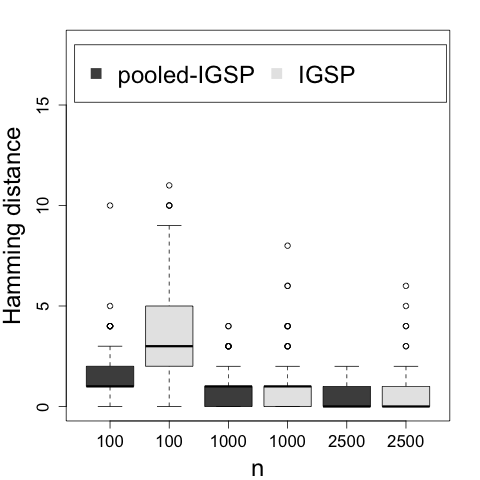}
}
\subfigure[]{
\includegraphics[scale=0.3]{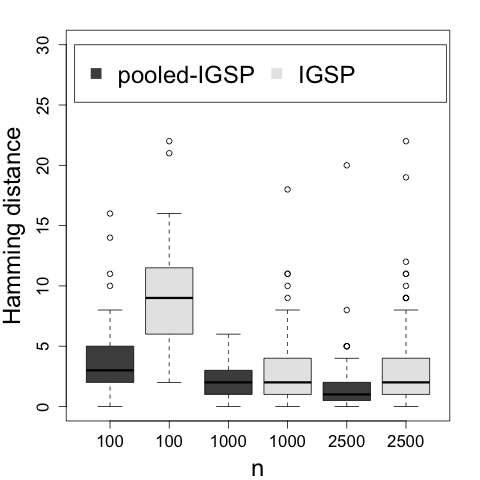}
}
\subfigure[]{
\includegraphics[scale=0.3]{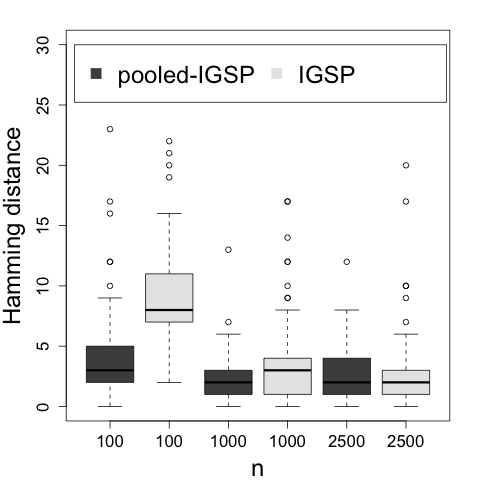}
}
\subfigure[]{
\includegraphics[scale=0.3]{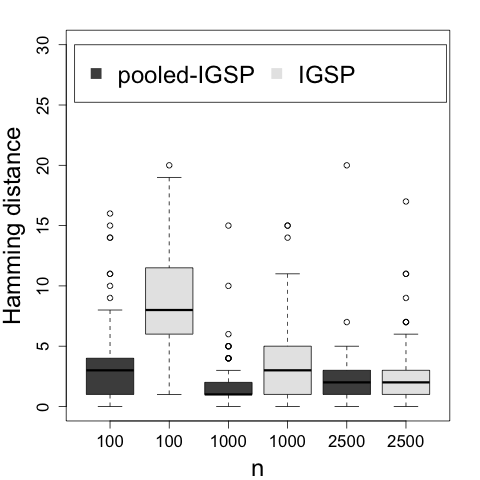}
}
\caption{Distributions of Hamming distances of recovered DAGs using IGSP and a heuristic pooled-IGSP for 10-node graphs under (a) perfect, (b) imperfect, and (c) inhibitory interventions and 20-node graphs under (d) perfect, (e) imperfect, and (f) inhibitory interventions}
\label{fig:sim_pool_extra}
\end{figure*}


\end{document}